\documentclass[a4paper,12pt]{article}

\usepackage{geometry,amssymb,amsmath,latexsym,amsthm}
\usepackage[utf8]{inputenc} 
\usepackage{setspace} 
\usepackage{subcaption}
\usepackage{pifont}
\usepackage{empheq} 
\usepackage[normalem]{ulem} 
\usepackage{enumitem} 
\usepackage{multirow} 
\usepackage{bbm} 
\usepackage{cases} 
\usepackage[final]{pdfpages} 
\usepackage{bigints} 
\usepackage{comment} 
\usepackage{booktabs} 
\usepackage{eurosym}
\usepackage[font=small,labelfont=bf]{caption}
\usepackage{indentfirst} 
\usepackage{chngpage} 
\usepackage{breqn} 
\usepackage{tabularx}

\usepackage[hidelinks]{hyperref}

\usepackage{thmtools}
\usepackage{cleveref}

\usepackage[sort]{natbib}

\usepackage{tikz}
\usetikzlibrary{snakes,perspective}
\usetikzlibrary{calc} 
\usetikzlibrary{shapes,arrows} 
\usetikzlibrary{arrows,positioning} 
\usepackage{pgfplotstable} 
\usepgfplotslibrary{fillbetween} 

\usepgflibrary{arrows} 
\usepgflibrary[arrows] 
\usetikzlibrary{arrows} 
\usetikzlibrary[arrows] 

\usepackage{graphicx} 
\usepackage{graphics}

\newtheorem{claim}{Claim}
\newtheorem{fact}{Fact}

\newtheoremstyle{exampstyle}
  {} 
  {} 
  {} 
  {} 
  {\bfseries\color{blue}} 
  {.} 
  {.5em} 
  {} 
\theoremstyle{exampstyle}\newtheorem{thm}{Theorem}
\theoremstyle{exampstyle}\newtheorem*{thm*}{Theorem}
\theoremstyle{exampstyle}\newtheorem{defn}{Definition}     
\theoremstyle{exampstyle}     
\theoremstyle{exampstyle}\newtheorem{lem}{Lemma}  
\theoremstyle{exampstyle}        
\theoremstyle{exampstyle}\newtheorem{prop}{Proposition}
\theoremstyle{exampstyle}

\theoremstyle{exampstyle}  
\theoremstyle{exampstyle}\newtheorem{Assumption}{Assumption}
\theoremstyle{exampstyle}

\theoremstyle{definition}

\theoremstyle{exampstyle}
\newtheorem{remarktmp}{Remark}
\newenvironment{remark}
	{ \begin{remarktmp} 	}
	{ 
		\medskip\hfill
		\end{remarktmp} 
	}

\allowdisplaybreaks[1]

\newtheorem{exampletmp}{Example}

\allowdisplaybreaks[1]
\newcommand{\xmark}{\ding{55}}%

\newtheoremstyle{bfnoteonly}%
{}{}%
{\itshape}{}%
{\bfseries \color{red}}{.}%
{ }%
{\thmnote{#3}}

\theoremstyle{bfnoteonly}
\newtheorem*{axxm}{}

\renewcommand{\epsilon}{\varepsilon}

\newcommand{\ChoiProb}{\rho}

\DeclareSymbolFontAlphabet{\amsmathbb}{AMSb}

\DeclareMathOperator*{\argmax}{argmax} 

\title{Frame-Dependent Random Utility\thanks{We thank Jose Apesteguia, Miguel Ballester, Peter Caradonna, Chris Chambers, Kyle Hyndman, Ryo Kambayashi, Kohei Kawaguchi, Jay Lu, Paola Manzini, Marco Mariotti, Francesca Molinari, Hiroki Nishimura, Kota Saito, Ke Shaowei, Chris Turansick, Ece Yegane for their comments.}}
\author{Paul H.Y. Cheung\thanks{Naveen Jindal School of Management, University of Texas at Dallas. Email: paul.cheung@utdallas.edu}  \ \   \  \ Yusufcan Masatlioglu\thanks{Department of Economics, University of Maryland. Email: yusufcan@umd.edu}}
\date{Jan 31, 2025}

\begin{document}
\maketitle

\begin{abstract}

We explore the influence of framing on decision-making, where some products are framed (e.g., displayed, recommended, endorsed, or labeled). We introduce a novel choice function that captures observed variations in framed alternatives. Building on this, we conduct a comprehensive revealed preference analysis, employing the concept of frame-dependent utility using both deterministic and probabilistic data. We demonstrate that simple and intuitive behavioral principles characterize our frame-dependent random utility model (FRUM), which offers testable conditions even with limited data. Finally, we introduce a parametric model to increase the tractability of FRUM. We also discuss how to recover the choice types in our framework.


\end{abstract}

\section{Introduction}

Framing is a powerful marketing tool that shapes consumer choices by altering how existing or additional information about a product is conveyed. Some framing techniques reframe the same information in different terms—for instance, describing a product as ``25\% fat'' versus ``75\% lean'' directs consumer attention to distinct aspects, subtly influencing evaluation without changing the product itself. Other frames, such as endorsements, recommendation displays, or ``best-seller'' banners, offer subjective cues that consumers may interpret as relevant to their preferences, thereby enhancing perceived value. In this paper, we focus on this latter form of framing. However, our findings are broadly applicable to other types of framing as well. 

Marketers and practitioners employ framing techniques to shape consumer perceptions and behaviors, ultimately aiming to enhance sales, foster brand loyalty, and capture greater market share. These include but are not limited to traditional window displays (\citealp{sen2002}), virtual presentations (\citealp{González2021625}), recommendation signage (bestsellers or editor's picks) (\citealp{Goodman2013HelpDifficulty}), platform endorsements (Amazon's Choice) (\citealp{bairathi2024value}), and quantified claims (\citealp{Yu2023}). The evidence for these simple modifications altering choice behavior is conclusive and indisputable across a wide spectrum of economic activities (e.g., for the labor market, \citealp{Horton2017TheExperiment}; for hospitality and tourism, \citealp{Litvin2008ElectronicManagement}; for music streaming services, \citealp{Li2007AFeatures,Adomavicius2018EffectsPay}; for e-commerce of commodities or goods, \citealp{Senecal2004TheChoices,Haubl2000ConsumerAids,Rowley2000ProductBots,Vijayasarathy2000PrintIntentions,Goodman2013HelpDifficulty}).

Framing is designed to influence consumer decisions, aiming to nudge individuals towards particular choices or behaviors. However, these effects are not uniform across all consumers; different individuals respond to marketing stimuli in unique ways due to their preferences, attention, or other factors. Moreover, these strategies can lead rational consumers to adjust their behavior, making it difficult to discern whether observed subsequent behaviors are entirely rational in the presence of framing.

In real life, consumer data is typically available in aggregate form rather than as individual data. This poses additional challenges for analyzing the impact of marketing strategies on decision-making at the individual level. While aggregate data can reveal overall trends—such as increased purchases after a promotion—it lacks the granularity to show how each consumer was specifically influenced. For instance, we might see a spike in sales but cannot determine whether certain groups were more influenced by these strategies. Without individual-level data, identifying who was most affected by marketing tactics or how they adjusted their decision-making process in response becomes highly challenging. Therefore, while aggregate data can indicate the overall impact of a marketing campaign, it is not clear whether it can provide insight into the extent or direction of influence on individual consumers.

An important question arises: Is it possible to infer individual behavior from aggregate data when marketing strategies are in play? In theory, aggregate data represents a collection of choices made by individuals, each with their preferences and susceptibilities to influence. Distinguishing between influenced and uninfluenced consumers remains difficult. In this paper, we aim to reveal whether (i) framing affects consumers' evaluations, and (ii) this influence can align with the utility maximization paradigm.

To tackle the question, we first extend traditional rational choice theory in deterministic domain by introducing the concept of frame-dependent utility, which posits that the framing of an option can enhance its perceived utility. To elucidate this concept, consider two scenarios: (i) $x$ is chosen when $y$ is presented as the framed option, and (ii) $z$ is chosen when nothing is framed. These choices challenge the frame-dependent utility hypothesis since, in the first scenario, the non-framed $x$ appears to outperform all unframed alternatives, contradicting the selection made in the second scenario. Therefore, in this paper, we aim to uncover the conditions under which choices influenced by frames can be seen as the result of utility maximization. Our investigation spans both deterministic and stochastic domains. We introduce behavioral postulates that rely solely on variations in the framed set, providing a pragmatic approach for testing the rationality hypothesis.

In the next section, we introduce our framework, presenting a model within a deterministic setting that we name the Frame-Dependent Utility Model (FUM). We demonstrate that FUM can be succinctly characterized by a behavioral postulate called \textit{Independence of Irrelevant Framed Alternatives} (IIFA). IIFA states that removing the frame from an alternative does not affect the choice if the alternative was not initially chosen. In the above example, $y$ was framed in the first scenario but not in the second. Hence, IIFA implies that the choice must remain the same. This principle parallels (yet different from) the axiom of Independence of Irrelevant Alternatives (IIA), a fundamental principle in classical choice theory that posits that choices should remain unaffected by non-selected options.\footnote{We should note that the standard IIA does not capture this since the chosen option here is the non-framed version of itself; for a detailed explanation, see Section \ref{sect: Deterministic}.} IIFA is a necessary and sufficient condition for utility maximization in the presence of framing.

Once we establish the result in the deterministic domain, we are ready to tackle whether we can infer individual behavior from aggregate data. In Section~\ref{sec:probabilistic}, we study the implications of frame-dependent utility on aggregate behavior as observed through probabilistic choice data. Our primary question concerns identifying the impact of individual preference maximization on overall market behavior. Specifically, how do we determine if aggregate behavior aligns with the hypothesis that it arises from a group of preference-maximizing consumers? To illustrate, consider a scenario where product $a$ is selected 70\% of the time when nothing is framed, 45\% when either product $b$ or $c$ is framed, and only 10\% when both $b$ and $c$ are framed. Note that the market share of $a$ decreases as the number of framed options increases. It seems like a case where utility maximizers are more drawn to the framed options than the unframed option $a$. Hence, it is tempting to conclude that this example can be rationalized by utility maximization with a heterogeneous population. However, we will demonstrate that such a conjecture is incorrect, i.e., these data cannot be represented as the aggregate behavior of preference-maximizing consumers. This example illustrates that identifying the implications of aggregate behavior among utility-maximizing individuals is not a straightforward task.

We first define the Frame-dependent Random Utility Model (FRUM). Similar to the classic Random Utility Model (RUM), FRUM assumes that each agent makes deterministic choices based on an underlying (frame-dependent) utility. The foundational RUM’s empirical validity has been explored by \citet{Falmagne1978}, \citet{Barbera1986FalmagneOrderings}, and \citet{McFadden1990StochasticPreference}, who demonstrate that the non-negativity of the Block-Marshak polynomials (BM) serves as a criterion for the choice data’s consistency with RUM. In our environment, the classical BM can be defined using framed-set variation for framed alternatives. However, it is silent on how non-framed alternatives should behave (see the above example). Therefore, we introduce new conditions for non-framed alternatives and show that they characterize FRUM. Our conditions delineate the behavioral content of heterogeneous frame-dependent utility maximizers.

Lastly, we investigate frame-dependent utility in a probabilistic environment using a parametric model, which we call the Frame-dependent Utility Luce Model (F-Luce). The model has a finite set of alternative-specific parameters. We show that two simple behavioral postulates can fully characterize this model. The first postulate, called Strong Luce-IIA, mimics the famous Luce-IIA (\citealp{luce1959individual}) in our framework. The axiom requires that as long as the status of two alternatives does not change across two frames, the choice ratio must remain constant. Indeed, Luce-IIA, when adapted to this framework, is implied by Strong Luce-IIA, since the choice ratio between framed items remains constant. The second postulate is F-Regularity, which states that a non-framed alternative must be chosen more often as the framed set becomes smaller. Our characterization results do not require knowledge of the full data. We show that if the choice data include all framed sets of size two or fewer, the characterization holds. In addition, we can uniquely pin down the utility weights up to a scaling factor.

The analysis of the framing effect from a choice-theoretic perspective can be traced back to the work of \citet{Salant2008}, who proposed a comprehensive framework for studying the framing effect. Building on this foundation, more recently, \citet{Galperti-CerigioniJEEA2023} examined how the sequence in which attribute information about alternatives is presented can influence decision-making, and \citet{BHATTACHARYA2021102553} investigated the impact of frame on attention. Furthermore, \citet{Goldin-ReckJPE2020} narrowed their focus to two alternatives, employing statistical methods and assumptions to facilitate identification. On the other hand, scholars have also been drawn to explore different classes of random utility models (e.g. \citealp{Gul_Pesendorfer2006,jose_Miguel_Lu_SCRUM2017,APESTEGUIA2023105674,MANZINI2018162,Turansick_2024,Frick_Iijima_Strzalecki2019}). To the best of our knowledge, our paper is the first choice-theoretic work to address the issue of frame-dependent random utility. Meanwhile, several different strands of research have augmented or departed from choice-set variation in the standard model. For example, some studies have employed list variation to examine choices (\citealp{Ishii2021ALists, Guney2014ALists}) and approval rates (\citealp{Manzini2021SequentialBehaviour}). \citet{Dardanoni_JPE2023} proposed using mixture choice data and \citet{Azrieli_Rehbeck2022} studied marginal frequencies of choice and availability but not standard choice probabilities in each menu. \citet{Cheung-Masatlioglu2024} and \citet{Ke2021LearningBox} investigated how recommendations can affect people's choices and beliefs, and \citet{Natenzon2019RandomLearning} and \citet{Guney2018Aspiration-basedChoice} explored the impact of unselectable ``phantom'' options on decision-making. These lines of research, in a spirit similar to our model, enrich the standard choice environment and deepen our understanding of human behavior.

\section{Framework and Micro Foundation}\label{sect: Deterministic}

Let $X$ be a finite set of alternatives (e.g., all documentaries available on Netflix or all 65-inch Smart TVs sold on Amazon). We let $F \subseteq X$ denote the set of framed alternatives. When confusion is unlikely, we also call $F$ the frame. In this framework, a decision problem describes not only which alternatives are available, denoted by $S$, but also which alternatives are framed, denoted by $F$. Therefore, $(F,S)$ constitutes a decision problem where $F,S \subseteq X$ and $S \neq \emptyset$. One can define a general choice function $c$ such that $c(F,S) \in S$, where $S$ represents the set of feasible alternatives and $F$ represents the framed alternatives. There are two types of variation in the general model: Frame variation and feasible set variation. In other words, this framework is more ``data-hungry'' than the standard choice model since the analyst must observe choices from all decision problems.

In this paper, we address the challenge of extensive data requirements by assuming that the set of available alternatives remains fixed, $X$. Consequently, our model operates without variation in the availability of alternatives, reflecting environments like Amazon, where product availability exhibits minimal or no fluctuation. Under this assumption, any subset of 
$X$, denoted as $F$, can represent a decision problem in which all alternatives within $F$ are framed and the rest are not framed. Importantly, the empty set also constitutes a valid decision problem, representing a scenario where no alternatives are framed.

While frames can influence choices, they do not constrain them. To capture this, we define a choice rule $c$ as a function of the frame, $F$, but allow for $c(F)$ to be outside of $F$. Hence, the only restriction we impose is $c(F) \in X$. Notice that our choice function differs from the one used in the choice theory literature, where $c(F)$ is always in $F$. This distinction will be important when we introduce behavioral properties. We denote the classical choice function by $d$, where $d(S) \in S$ for all $S \subset X$, and $S$ represents the set of available options. Additionally, we allow $c$ to be observable only for some frames but not others. This assumption aims to capture real-world environments in which the collection of the frame is just a fraction of the entire product space. Let $\mathcal{D} \subseteq 2^X$ denote all possible frames for which we have data. For instance, $\mathcal{D}$ can include all frames of size $1$. The following definition captures the choice rule under this framework. Depending on the results we provide in this paper, we highlight the necessary requirements in the domain.

\begin{defn}
A deterministic choice rule $c$ on $\mathcal{D}$ is a mapping from $\mathcal{D}$ to $X$ such that $c(F) \in X$ for all $F \in \mathcal{D}$.
\end{defn}

Consider that if a decision-maker (DM) consistently acts in alignment with a genuine underlying preference that is frame-independent, then framing should not influence decisions. This implies that the choice function $c(F)$ remains constant, irrespective of framing ($c(F)=c(F')$ for all $F,F'\in \mathcal{D}$). However, contrary to this assumption, existing evidence robustly demonstrates that framing significantly impacts decision-making. In this context, we introduce a model of frame-dependent utility maximization. This model adheres to the foundational principles of rational choice, where each option is assigned a subjective value, and the DM opts for the option with the highest value.

Framing influences decisions solely by enhancing the perceived value of the framed option. Consequently, framing positively alters the option's relative position within the DM's preferences. This concept aligns with traditional theories that negate the presence of ``context effects,'' asserting that an option's valuation is solely dependent on its framed status, without being influenced by other framed options.

To formalize this, suppose the DM considers two functions, $u$ and $v$, mapping from a set $X$ to the real numbers $\mathbb{R}$. For a framed product $x$ (where $x\in F$), its utility is perceived as the sum $u(x) + v(x)$; if not framed, the utility remains $u(x)$. The DM's goal is to select the option that maximizes this utility. The model assumes $v(x)$ is non-negative to represent positive framing effects. The choice function under this framework, denoted as $c_{(u,v)}$, is written as: \begin{align*}
    c_{(u,v)}(F)=\argmax_{x\in X} \ U_F(x) \ \ \   \text { where } \  \  U_F(y) =\begin{cases} u(y)+v(y) & \text{ if  } y \in F \\
    u(y) & \text{ if  } y \notin F
    \end{cases}  \tag{FUM}
\end{align*} where  $v(x) \geq 0$.\footnote{Since our domain is finite and choices are unique, we assume $U_F$ is an injective function.} We say a choice rule $c$ has a Frame-dependent Utility Model (FUM) representation on $\mathcal{D}$ if there exists $(u,v)$ with $v \geq 0$ such that  $c=c_{(u,v)}$. Note that framing does not affect utilities in an arbitrary way. First of all, as long as the status of the option is the same under $F$ and $F'$, its utility remains the same, i.e.,  if $x \notin F \cup F'$ or $x \in F \cap F'$ then $U_F(x)=U_{F'}(x)$. Moreover, the relative ranking of two alternatives does not change as long as they are not framed, i.e., if $x,y  \notin F \cup F'$ then $U_F(x) > U_{F}(y)$ if and only if $U_{F'}(x)>U_{F'}(y)$. The same applies if both of them are framed.  $\mathcal{C}_{\text{FUM}}$ denotes the collection of the choice functions that have a FUM representation.

Finally, we illustrate the distinctions between classical choice functions (\(d\)) and frame-dependent choice functions (\(c\)) within a straightforward environment \(X = \{a, b\}\) (see \Cref{Table:Deterministic}). This comparison is crucial for understanding the subsequent analysis and establishing a connection with their classical counterparts. A classical choice function is defined only on the non-empty subsets of \( X \). Since the chosen element must belong to the decision problem, the only non-trivial choice problem is \( \{a, b\} \). In this case, there are exactly two such functions, as shown in the first two columns of \Cref{Table:Deterministic}. On the other hand, our choice function is defined on any subsets of \( X \). Moreover, all choice problems are non-trivial; hence, there are $16$ different choice functions in our framework. Even in this simple environment, frame-dependent choice functions and classical choice functions differ significantly. The next section introduces conditions to restrict the possible choice functions in order to identify  $\mathcal{C}_{\text{FUM}}$.

\subsection{Behavioral Characterization}

In this section, we explore the behavioral implications of the model in a deterministic environment. The first behavioral postulate is the key axiom in this environment, which states that removing some of the unchosen alternatives from the frame does not influence the final choice. Note that it shares a similar flavor to the famous IIA axiom (\textit{i.e.}, Independence of Irrelevant Alternatives) in the classical choice theory literature.\footnote{This property is also known as Sen's $\alpha$ axiom (\citealp{Sen71}), Postulate 4 of \citet{Chernoff54}, C3 of \citet{Arrow1959}, the Heritage property of \citet{aizerman1995theory}, or the Heredity property of \citet{aleskerov2007utility}.} We call this the Independence of Irrelevant Framed Alternatives, abbreviated as IIFA.

\begin{axxm}[IIFA]\label{Axiom:IIFA}
If $c(F) \notin F \setminus F'$ and $F'\subseteq F$, then $c(F)=c(F')$.
\end{axxm}

Since this axiom looks similar to the standard IIA, one might think they are equivalent. While the standard IIA applies to classical choice functions, in which the winner always belongs to the choice set, our postulate applies to new choice objects, in which the winner can be a framed or non-framed item. This seemingly small distinction has important implications. In fact, if we include ``$c(F)$ belonging to $F'$'' in the premise of IIFA, then IIFA becomes the standard IIA, which is stated below as IIFA(1). However, IIFA(1) is strictly weaker than IIFA (see \Cref{Table:Deterministic}).

\begin{axxm}[IIFA(1)]\label{Axiom:IIFA(1)} 
If $c(F) \in F'$ and $F'\subseteq F$, then $c(F)=c(F')$.
\end{axxm}

In other words, seemingly similar axioms have distinct implications in our framework. Hence, the existing results in the literature on the standard framework might not be valid in our framework. Since IIFA(1) is weaker than IIFA, we would like to state the counterpart of IIFA(1). To this end, we include ``$c(F)$ is not in $F$'' in the premise and call it IIFA(2). It states that the winner outside of the frame will persist when fewer alternatives are framed.

\begin{axxm}[IIFA(2)]\label{Axiom:IIFA(2)}
If $c(F) \notin F$ and $F'\subseteq F$, then $c(F)=c(F')$.
\end{axxm}

It is easy to see that the combination of IIFA(1) and IIFA(2) is equivalent to IIFA. Notice that the model satisfies IIFA. To see this, first consider IIFA(1), so $c(F) \in F'$. The model states that $c(F)$ is evaluated according to the framed version, and it must be better than all other framed and non-framed alternatives. Therefore, the evaluated ranking of $c(F)$ remains the best alternative even if some alternatives are ``downgraded.'' Next, consider the second postulate, IIFA(2), so $c(F) \notin F$. The model states that the non-framed item $c(F)$ is already better than the framed items in $F$ and other non-framed items. Hence, the non-evaluated version of $c(F)$ is still superior to other alternatives when some other alternatives are not evaluated anymore.

\Cref{Table:Deterministic} helps to understand the implications of IIFA(1) and IIFA(2) for both classical choice functions ($d$) and frame-dependent choice functions ($c$). Since all classical choice functions (\( d_1 \) and \( d_2 \)) satisfy IIFA(1), Hence, the classical utility maximization has no empirical content for two alternatives.\footnote{IIFA(2) is vacuously true for the classical choice functions.} In contrast, IIFA(1) and IIFA(2) have empirical content even in this simple environment. Notably, \( c_1 \) satisfies both IIFA(1) and IIFA(2). These axioms are independent, as demonstrated by \( c_4 \), which satisfies only IIFA(1), and \( c_5 \), which satisfies only IIFA(2). Finally, \( c_6 \) does not satisfy either of these axioms (IIFA(1) is violated since \(c_6(\{a,b\})\neq c_6(\{a\})\), and IIFA(2) is violated since \(c_6(\{a\})\neq c_6(\emptyset)\)).\footnote{\Cref{Table:Deterministic} only includes frame-dependent choice functions with $c(\{a,b\})=a$. Hence there are eight more choice functions with $c(\{a,b\})=b$ (16 in total). 6 out of 16 frame-dependent choice functions satisfy both IIFA(1) and IIFA(2).}

\begin{table}[h!]
\begin{center}
\begin{tabular}{c|cc|cccccccc}
\toprule
           & $d_1$ & $d_2$ & $c_1$ & $c_2$ & $c_3$ & $c_4$ & $c_5$ & $c_6$ & $c_7$ & $c_8$ \\
\midrule
$\{a,b\}$  & $a$   & $b$   & $a$   & $a$   & $a$   & $a$   & $a$   & $a$   & $a$   & $a$   \\
$\{a\}$    & $a$   & $a$   & $a$   & $a$   & $a$   & $a$   & $b$   & $b$   & $b$   & $b$   \\
$\{b\}$    & $b$   & $b$   & $a$   & $b$   & $b$   & $a$   & $b$   & $a$   & $a$   & $b$   \\
$\emptyset$& -     & -     & $a$   & $a$   & $b$   & $b$   & $b$   & $a$   & $b$   & $a$   \\
\midrule
IIFA(1) & \checkmark     & \checkmark    & \checkmark   & \checkmark    & \checkmark   & \checkmark    &  \xmark  & \xmark   & \xmark  & \xmark   \\
IIFA(2) & -     & -     & \checkmark   & \checkmark    & \checkmark   & \xmark    &  \checkmark  & \xmark   & \xmark  & \xmark   \\
\bottomrule
\end{tabular}
\end{center}
\caption{Classical Choice Functions ($d$) versus Frame-Dependent Choice Functions ($c$). $d_i$'s are only defined on the non-empty subsets and $d_i(S)$ must belong to $S$, e.g., $d_1(\{a\})=d_2(\{a\})=a$, hence the only non-trivial choice is made when the decision problem is $\{a,b\}$. On the other hand, $c_i(F)$ could be outside of $F$, such as $c_2(\{b\})=a$. $c_1$-$c_3$ satisfy both IIFA(1) and IIFA(2). While $c_4$ only satisfies IIFA(1),  $c_5$ only satisfies IIFA(2). Finally, $c_6$-$c_8$ do not satisfy any of these axioms.}
\label{Table:Deterministic}
\end{table}

Our next result demonstrates that the full behavioral implications of FRUM are encapsulated by IIFA(1) and IIFA(2). Hence, \( c_1-c_3 \) are the only choice functions in \Cref{Table:Deterministic} that have FRUM representations.

\begin{thm}[Characterization] \label{Thm: Deterministic} Let $\mathcal{D}$ includes all frames with $|F|\leq 3$. Then, $c$ has an FUM representation  if and only if $c$ satisfies \nameref{Axiom:IIFA}.\end{thm}

We now discuss the identification and uniqueness properties of our model. Note that our domain is finite; hence, we can only attain ordinal uniqueness. Assume that $(u_1, v_1)$ and $(u_2, v_2)$ represent the same choice rule. Notice that when there is no framing, alternatives are evaluated according to $u_i$. Therefore, the best alternative in $X$ according to both $u_1$ and $u_2$ must be the same. Additionally, if $x$ is different from $c(\emptyset)$ and is chosen when it is framed, its evaluated utility $u_i + v_i$ is revealed to be higher than the best alternative in $X$ according to $u_i$ (revealed preference). Moreover, if two different framed alternatives are chosen on two occasions, one can uniquely determine their ranking in $u_i + v_i$. Lastly, the lower contour set of the best alternative in $X$ must be the same.

\begin{prop}[Uniqueness]\label{Thm: PFdeterministic uniqueness}
Let $(u_1,v_1)$ and $(u_2,v_2)$ be two FUM representations of the same choice data $c$. Then, for $x,y \in X$ \\
i) $\argmax_{x \in X} u_1(x)=\argmax_{x \in X} u_2(x):=a$, \\
ii) $u_1(a)>(u_1 + v_1)(x) $ if and only if $ u_2(a)>(u_2 + v_2)(x)$, and\\
iii) $(u_1 + v_1)(x) > (u_1 + v_1)(y) >u_1(a)$ if and only if $(u_2 + v_2)(x) > (u_2 + v_2)(y) >u_2(a)$.
\end{prop} 

This proposition also states that without framing, the choice is unique, and the ordinal rankings are uniquely identified as long as it matters for the choice, i.e., $c_{(u_1,v_1)}(F)=c_{(u_2,v_2)}(F)$ for all $F$.

\section{Random Utility}\label{sec:probabilistic}

In the preceding section, we presented the groundwork for our model by illustrating how framing influences deterministic choices. This section advances the discussion to encompass probabilistic choices. Probabilistic choice data can be interpreted in two primary ways. The first, an interpersonal perspective, suggests that variability in choice data stems from the heterogeneity of individual types within a group, with only their collective behavior observable to an external observer. The second, an intrapersonal perspective, reflects the variability in the choices of a single individual across different contexts. Our model accommodates both interpretations, offering insights into intrapersonal and interpersonal probabilistic choices.

Probabilistic choices have been extensively analyzed within the traditional choice-set variation framework. However, as we move on to studying frame-set variation, the usual probabilistic choice function will fall short of capturing the essence of choosing outside of the frame $F$. To address this limitation and integrate the concept of variation in framed sets, we expand the definition of a probabilistic choice rule, paralleling the approach taken in our analysis of deterministic choices.

\begin{defn}
A choice rule $\rho$ is a mapping from $X \times \mathcal{D}$ to $[0,1]$ such that $\sum\limits_{x \in X} \rho(x,F)=1$.
\end{defn}

As in deterministic choice, $\rho$ differs from the classical probabilistic choice $\pi$. Formally,  $\pi$ is also a mapping from $X \times \mathcal{D}$ to $[0,1]$ such that $\sum_{x \in X} \rho(x,F)=1$. Yet, it is also assumed $\pi(x, F) = 0$ for $x \notin F$, whereas $\rho(x, F)$ could be strictly positive.

In this section, we investigate the behavioral implications of frame-dependent random utility. Similar to the standard Random Utility Model (RUM), it is assumed that the observed data are generated by a group of diverse individuals making deterministic choices. Before detailing our model, we revisit the classical random utility model. RUM accounts for a ``rational'' heterogeneous population, with each entity or type within this population maximizing its utility. Let $\mu$ represent a probability distribution over the set of rational choice functions, denoted by $\mathcal{C}_{\text{R}}$.\footnote{Formally, $\mathcal{C}_{\text{R}}=\{d | d(S) = \argmax_{x \in S} u(x) \text{ for some (injective) utility function } u \text{ for all } S \in \mathcal{D}.\}$} This probability distribution, $\mu$, constitutes a probabilistic choice function $\pi_{\mu}$, which is defined as follows: \begin{equation}\label{(RUM)}
      \pi_{\mu}(x,S)=\mu( \{ d \in \mathcal{C}_{\text{R}} | \ x=d(S)  \})  \tag{RUM}
  \end{equation}

The probability of $x$ being chosen from $S$ is the frequency of those types who choose $x$ from $S$. We say that a probabilistic choice function $\pi$ has a Random Utility Model (RUM) representation if there exists $\mu$ such that $\pi = \pi_{\mu}$.

Therefore, to study frame-dependent random utility, we assume instead that each decision-maker in the population belongs to our deterministic model (FUM). Let $\mu$ be a probability distribution over $\mathcal{C}_{\text{FUM}}$. The probability distribution $\mu$ constitutes a probabilistic choice function $\rho_{\mu}$ such that \begin{equation}\rho_\mu(x,F)=\mu( \{ c \in \mathcal{C}_{\text{FUM}} | \ x=c(F)  \})\tag{FRUM} \end{equation} for every $F\in \mathcal{D}$.  We say that a probabilistic choice function, $\rho$, has a Frame-Dependent Random Utility  Model (FRUM) representation if there exists $\mu$ on $\mathcal{C}_{\text{FUM}}$ such that $\rho=\rho_{\mu}$. FRUM tests the hypothesis of a group of utility-maximizing individuals affected by frames. In FRUM, the only departure from RUM we introduce here is that frame can change the ranking of framed items. Note that FRUM has a much richer type space compared to RUM.\footnote{ For $|X|=n$, RUM has $n! (=|\mathcal{C}_R|)$ choice types, while FRUM has $\sum_{i=1}^n i! {n \choose i}i (=|\mathcal{C}_{\text{FUM}}|)$ choice types. For \(n =[2,3,4,5,6]\), choice types in RUM and FRUM are \([2,6,24,120,720]\) and \([6,33,196,1305,9786]\), respectively.}

\subsection{RUM versus F-RUM}

In this section, we compare our model with the classical random utility model. Figure \ref{fig:FRUM} illustrates the relationship between F-RUM and RUM on the Marschak-Machina triangle in a domain of three alternatives, $X = \{x, y, z\}$. This is the simplest environment we can utilize to study the differences and similarities between these models. 

The Marschak-Machina triangle is a tool used in decision theory to represent stochastic choices. Each corner represents a different alternative. In RUM, there are only four non-trivial choice problems describing the available alternatives. For example, $\pi(\cdot, \{x, y\})$ assigns zero choice probability to $z$ since $z$ is not available. In F-RUM, there are eight non-trivial observations ($=2^{|X|}$) describing framed alternatives. For example, $\rho(\cdot, \{x, y\})$ assigns a positive choice probability to all alternatives, including the non-framed alternative $z$.

\begin{figure}[h!]
   \begin{subfigure}[b]{0.5\textwidth}
        \centering

\begin{tikzpicture}[line cap=round,line join=round,      3d view={135}{45},                    scale=4.5]


\def\ux{3}   
\def\uy{2}   
\def\uz{1}   

\def\vx{1}   
\def\vy{1}   
\def\vz{1}   

\def\xx{1}   
\def\yy{1}   
\def\zz{1}   

\pgfmathsetmacro\pxxy{(\ux+\vx)/(\ux+\vx+\uy+\vy)} 
\pgfmathsetmacro\pyxy{(\uy+\vy)/(\ux+\vx+\uy+\vy)} %

\pgfmathsetmacro\pxxz{(\ux+\vx)/(\ux+\vx+\vz+\uz)} 
\pgfmathsetmacro\pyxz{(0)/(\ux+\vx+\vz+\uz)} %

\pgfmathsetmacro\pxyz{(0)/(\vy+\uy+\vz+\uz)} 
\pgfmathsetmacro\pyyz{(\uy+\vy)/(\vy+\uy+\vz+\uz)} %

\pgfmathsetmacro\pxX{(\ux+\vx)/(\ux+\vx+\vy+\uy+\vz+\uz)} 
\pgfmathsetmacro\pyX{(\uy+\vy)/(\ux+\vx+\vy+\uy+\vz+\uz)} %

\pgfmathsetmacro\pzxy{\zz*(1-\pxxy/\xx-\pyxy/\yy)}
\pgfmathsetmacro\pzxz{\zz*(1-\pxxz/\xx-\pyxz/\yy)}
\pgfmathsetmacro\pzyz{\zz*(1-\pxyz/\xx-\pyyz/\yy)}
\pgfmathsetmacro\pzX{\zz*(1-\pxX/\xx-\pyX/\yy)}

\coordinate (O)  at (0,0,0);
\coordinate (CE)  at (1/3,1/3,1/3);
\coordinate (A)  at (\xx,0,0);
\coordinate (B)  at (0,\yy,0);
\coordinate (C)  at (0,0,\zz);
\coordinate (Pxy)  at (\pxxy,\pyxy,\pzxy);
\coordinate (Pxz)  at (\pxxz,\pyxz,\pzxz);
\coordinate (Pyz)  at (\pxyz,\pyyz,\pzyz);
\coordinate (PX)  at (\pxX,\pyX,\pzX);
\coordinate (Qxy)  at (\pxxy,\pzxy,\pyxy);
\coordinate (Wxy)  at (\pzxy,\pyxy,\pxxy);
\coordinate (Qxz)  at (\pyxz,\pxxz,\pzxz);
\coordinate (Wxz)  at (\pxxz,\pzxz,\pyxz);
\coordinate (Qyz)  at (\pzyz,\pyyz,\pxyz);
\coordinate (Wyz)  at (\pyyz,\pxyz,\pzyz);

\path [name path=Pxz--Qxz] (Pxz) -- (Qxz);
\path [name path=Pxz--Wxz] (Pxz) -- (Wxz);
\path [name path=Pxy--Wxy] (Pxy) -- (Wxy);
\path [name path=Pxy--Qxy] (Pxy) -- (Qxy);
\path [name path=Pyz--Qyz] (Pyz) -- (Qyz);
\path [name path=Pyz--Wyz] (Pyz) -- (Wyz);
\path [name intersections={of=Pxz--Qxz and Pxy--Qxy,by=E1}];
\path [name intersections={of=Pxz--Qxz and Pxy--Wxy,by=E2}];


\draw[red!50!black,thin,fill=red!50,fill opacity=0.5] (E1) -- (E2) -- (Pxy) -- cycle;


\draw[gray, dotted, very thin] (Qyz) -- (Pyz);
\draw[gray, dotted, very thin] (Wyz) -- (Pyz);
\draw[gray, dotted, very thin] (Qxz) -- (Pxz);
\draw[gray, dotted, very thin] (Wxz) -- (Pxz);

\draw[gray, dotted, very thin] (Qxy) -- (Pxy);
\draw[gray, dotted, very thin] (Wxy) -- (Pxy);
\draw[gray!50!black,thick,fill=gray!50,fill opacity=0.1] (A) -- (B) -- (C) -- cycle;
\fill (A)  circle (.2pt) node[below left] {$x$};
\fill (B)  circle (.2pt) node[below right] {$y$};
\fill (C)  circle (.2pt) node[above] {$z$};

\fill[red] (2/4-.05,1/4+.05,1/4)  node {$\pi(xyz)$};

\fill[black] (Pxy)  circle (.3pt) node[below] {$\pi(xy)$};
\fill[black] (Pyz)  circle (.3pt) node[right] {$\pi(yz)$};
\fill[black] (Pxz)  circle (.3pt) node[left] {$\pi(xz)$};

\end{tikzpicture} 
      \caption{RUM}
    \end{subfigure}
   \begin{subfigure}[b]{0.5\textwidth}
        \centering
\begin{tikzpicture}[line cap=round,line join=round,             3d view={135}{45},        scale=4.5]

\coordinate (O)  at (0,0,0);
\coordinate (CE)  at (1/3,1/3,1/3);
\coordinate (A)  at (1,0,0);
\coordinate (B)  at (0,1,0);
\coordinate (C)  at (0,0,1);
\coordinate (Pe)  at (1/3,1/3,1/3);
\coordinate (Px)  at (.7,.2,.1);
\coordinate (Py)  at (.1,.7,.2);
\coordinate (Pz)  at (.2,.2,.6);
\coordinate (Pxy)  at (.45,.5,.05);
\coordinate (Pxz)  at (.5,.075,.425);
\coordinate (Pyz)  at (.1,.45,.45);

\coordinate (E1)  at (.2,.325,.475);
\coordinate (E2)  at (.425,.325,.25);
\coordinate (E3)  at (.2,.55,.25);

\draw[blue!50!black,thin,fill=blue!80,fill opacity=0.1] (E1) -- (E2) -- (E3) -- cycle;

\coordinate (F1)  at (.25,.45,.3);
\coordinate (F2)  at (.275,.45,.275);
\coordinate (F6)  at (.25,.325,.425);
\coordinate (F5)  at (.325,.25,.425);
\coordinate (F4)  at (.45,.25,.3);
\coordinate (F3)  at (.45,.275,.275);

\draw[red!50!black,thin,fill=red!50,fill opacity=0.5] (F1) -- (F2) -- (F3) -- (F4) -- (F5) -- (F6) -- cycle;

\draw[gray!50!black,thick,fill=gray!50,fill opacity=0.1] (A) -- (B) -- (C) -- cycle;
\fill (A)  circle (.2pt) node[below left] {$x$};
\fill (B)  circle (.2pt) node[below right] {$y$};
\fill (C)  circle (.2pt) node[above] {$z$};

\fill[white] (.5, .5,0)  circle (.0005pt) node[below] {$\rho(x)$};

\fill[black] (Px)  circle (.3pt) node[left] {$\rho(x)$};
\fill[black] (Py)  circle (.3pt) node[below] {$\rho(y)$};
\fill[black] (Pz)  circle (.3pt) node[above] {$\rho(z)$};
\fill[black] (Pxy)  circle (.3pt) node[right] {$\rho(xy)$};
\fill[black] (Pyz)  circle (.3pt) node[above] {$\rho(yz)$};
\fill[black] (Pxz)  circle (.3pt) node[below] {$\rho(xz)$};
\fill[red] (F5)  node[above]  {$\rho(xyz)$};
\fill[blue] (E3)  node[below] {$\rho(\emptyset)$};
\end{tikzpicture}
      \caption{FRUM}
    \end{subfigure}
  \caption{RUM versus F-RUM: This figure illustrates the differences and similarities between classical RUM and F-RUM for three alternatives, $X = \{x, y, z\}$. The left figure illustrates RUM. The red triangle identifies all possibilities for $\pi(\cdot, X)$ consistent with RUM, given three binary choices depicted on three edges of the simplex. The right figure illustrates F-RUM. Here, there are eight non-trivial observations describing framed alternatives. The right figure illustrates all possible choices for $F = \emptyset$ and $F = X$ given choices for binary and singleton frames. The red irregular hexagon identifies all possibilities for $\rho(\cdot, X)$ for F-RUM. Similarly, the blue triangle depicts all possibilities for $\rho(\cdot, \emptyset)$ for F-RUM.}
     \label{fig:FRUM}
\end{figure}
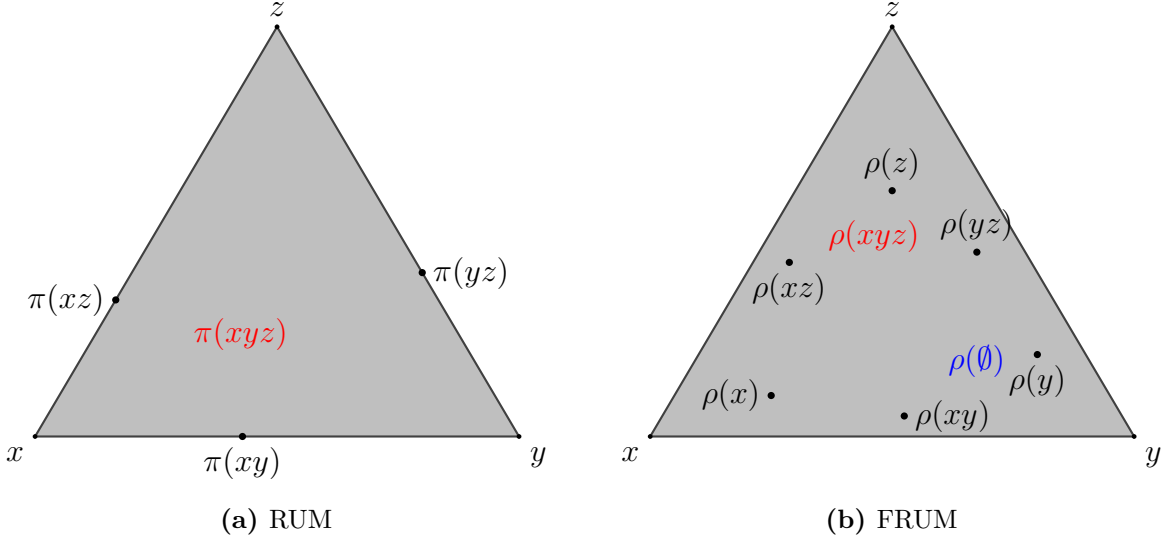

Figure \ref{fig:FRUM} depicts RUM and F-RUM. In each figure, choosing $x$ with probability one is represented by the left corner of the triangle. In RUM, stochastic choices for binary sets are fixed and denoted by black dots on the edges of the triangle. The distance between the corner $x$ and the black point labeled $\pi(xy)$ represents $\pi(y, \{x, y\})$, the choice probability of $y$ when $x$ and $y$ are the only alternatives. Similarly, the distance between the corner $y$ and $\pi(xy)$ represents $\pi(x, \{x, y\})$ in the same scenario. The points inside the Marschak-Machina triangle capture choice probabilities from the choice set $X$. The upside-down red triangle identifies all possibilities for $\pi(\cdot, X)$ for RUM given these binary choices. This figure illustrates both the predictive and explanatory power of RUM: If $\pi(\cdot, X)$ lies outside this triangle, the data cannot be explained by the random taste and utility maximization hypothesis.

The figure on the right panel shows the predictions for F-RUM given choices from binary and singleton frames. F-RUM restricts choices under both $F = \emptyset$ and $F = X$. The red irregular hexagon identifies all possibilities for $\rho(\cdot, X)$ consistent with F-RUM. Thus, this area captures the explanatory power of F-RUM given these choices under singleton and binary frames. F-RUM also predicts choices under $F = \emptyset$. The blue triangle identifies all possibilities for $\rho(\cdot, \emptyset)$ consistent with F-RUM. If either $\rho(\cdot, X)$ lies outside the red hexagon or $\rho(\cdot, \emptyset)$ lies outside the blue triangle, then the data cannot be explained by the utility maximization hypothesis with preferences influenced by frames. Note that the predictions of F-RUM are distinct from those of RUM.

\subsection{Hasse Diagram}

Within the RUM framework, \citet{Falmagne1978} addressed the question of whether individual frame-dependent utility maximization has any implications for aggregate data. For the characterization of RUM, \citet{Falmagne1978} utilized a well-known concept in the literature: The Block-Marschak polynomials, named after \citet{block1959random}'s seminal work on RUM. It has been shown that probabilistic choice data have a RUM representation if and only if their Block-Marschak polynomials are non-negative. \citet{block1959random} originally established the necessity, and \citet{Falmagne1978} demonstrated that the non-negativity of the polynomials is also sufficient.


In our framework, we also utilize Block-Marschak (BM) polynomials.  Let $q_\rho(a,F)$ denote the BM polynomials \textit{i.e.} for any $a \in F$, $$q_{\rho}(a,F):=\sum_{  B \supseteq F }(-1)^{|B \setminus F|} \rho(a,B)$$  The Block-Marschak polynomials are defined with respect to the choice data $\rho$. Throughout the paper, we mostly skip denoting $\rho$ and write $q(a,F)$ unless specified otherwise. $q_\rho$ is defined for each framed alternative.  We define 
auxiliary BM conditions for non-framed alternatives. As we shall see, they play an important role in FRUM. Formally, for $a \notin F$, $$ y_\rho(a,F) :=      \sum_{ a \notin B\supseteq F} (-1)^{|B \setminus  F|} \rho(a,B) $$ Again, we skip denoting $\rho$ and write $y(a,F)$ unless specified otherwise.

\begin{figure}[h!]
    \begin{center}
    \begin{tikzpicture}[ >=latex,glow/.style={%
    preaction={draw,line cap=round,line join=round,
    opacity=0.3,line width=4pt,#1}},glow/.default=yellow,
    transparency group]
            \tikzstyle{every node} = [rectangle]
         \node (s-1) at (1.5,-1.25) { };
        \node (s-2) at (0,-1.25) {};
        \node (s-3) at (-1.5,-1.25) {};
       \node (s) at (0,0) {$\emptyset$};
        \node (s1) at (3,2.5) {$\{a\}$};
        \node (s1-2) at (3,1.25) { };
        \node (s1-3) at (4.5,1.25) { };
        \node (s2) at (0,2.5) {$\{b\}$};
         \node (s2-1) at (-1.5,1.25) { };
        \node (s2-3) at (1.5,1.25) { };
        \node (s3) at (-3,2.5) {$\{c\}$};
        \node (s3-1) at (-4.5,1.25) { };
        \node (s3-2) at (-3,1.25) { };
        \node (s23) at (-3,5) {$\{b,c\}$};
        \node (s23-1) at (-4.5,3.75) {};
        \node (s13) at (0,5) {$\{a,c\}$};
        \node (s13-2) at (0,3.75) {};
        \node (s12) at (3,5) {$\{a,b\}$};
        \node (s12-3) at (4.5,3.75) {};
        \node (s123) at (0,7.5) {$\{a,b,c\}$};
        
     \foreach \i/\j/\txt/\p in {
      s/s-1/\text{$y(c,\emptyset)$}/above,
      s/s-2/\text{$y(b,\emptyset)$}/below,
      s/s-3/\text{$y(a,\emptyset)$}/above}
            \draw [red,dashed] (\i) -- node[pos=0.7,sloped,font=\tiny,\p] {\txt} (\j);

                \foreach \i/\j/\txt/\p in {
      s13/s13-2/\text{$y(b,\{a,c\})$}/above,
      s12/s12-3/\text{$y(c,\{a,b\})$}/above,
      s23/s23-1/\text{$y(a,\{b,c\})$}/above}
                  \draw [red,dashed] (\i) -- node[pos=0.7,sloped,font=\tiny,\p] {\txt} (\j);
            
                \foreach \i/\j/\txt/\p in {
      s3/s3-1/\text{$y(a,\{c\})$}/above,
      s3/s3-2/\text{$y(b,\{c\})$}/below,
      s2/s2-1/\text{$y(a,\{b\})$}/above,
      s2/s2-3/\text{$y(c,\{b\})$}/above,
      s1/s1-2/\text{$y(b,\{a\})$}/above,
      s1/s1-3/\text{$y(c,\{a\})$}/above}
                  \draw [red,dashed] (\i) -- node[pos=0.7,sloped,font=\tiny,\p] {\txt} (\j);
    
    \foreach \i/\j/\txt/\p in {
      s/s1/\text{$q(a,\{a\})$}/below,
      s/s2/\text{$q(b,\{b\})$}/below,
      s/s3/\text{$q(c,\{c\})$}/below}
            \draw [-] (\i) -- node[pos=0.7,sloped,font=\tiny,\p] {\txt} (\j);
         
         \foreach \i/\j/\txt/\p in {
      s1/s12/\text{$q(b,\{a,b\})$}/below,
      s3/s23/\text{$q(b,\{b,c\})$}/above}
            \draw [-] (\i) -- node[pos=0.5,sloped,font=\tiny,\p] {\txt} (\j);
            
             \foreach \i/\j/\txt/\p in {
      s1/s13/\text{$q(c,\{a,c\})$}/above,
      s2/s12/\text{$q(a,\{a,b\})$}/below,
      s2/s23/\text{$q(c,\{b,c\})$}/below,
      s3/s13/\text{$q(a,\{a,c\})$}/above}
            \draw [-] (\i) -- node[pos=0.8,sloped,font=\tiny,\p] {\txt} (\j);
         
          \foreach \i/\j/\txt/\p in {
      s12/s123/\text{$q(c,\{a,b,c\})$}/above,
      s13/s123/\text{$q(b,\{a,b,c\})$}/above,
      s23/s123/\text{$q(a,\{a,b,c\})$}/above}
            \draw [-] (\i) -- node[pos=0.5,sloped,font=\tiny,\p] {\txt} (\j);

\draw [glow=green] (s123) -- node[pos=0.5,sloped,font=\tiny] {} (s12);
\draw [glow=green] (s12) -- node[pos=0.5,sloped,font=\tiny] {} (s12-3);
\draw [glow=gray] (s23) -- node[pos=0.5,sloped,font=\tiny] {} (s3);
\draw [glow=gray] (s3) -- node[pos=0.5,sloped,font=\tiny] {} (s);
\draw [glow=gray,red,dashed] (s) -- node[pos=0.5,sloped,font=\tiny] {} (s-1);

\draw [glow=blue] (s123) -- node[pos=0.5,sloped,font=\tiny] {} (s23) ;
\draw [glow=blue] (s23) -- node[pos=0.5,sloped,font=\tiny] {} (s3);
\draw [glow=blue] (s3) -- node[pos=0.5,sloped,font=\tiny] {} (s);
\draw [glow=blue,red,dashed] (s) -- node[pos=0.5,sloped,font=\tiny] {} (s-1);
           \end{tikzpicture}    \end{center}
\caption{The new Hasse Diagram for FRUM consists of black solid lines representing the Block-Marschak polynomials (the old Hasse Diagram for RUM) and red dashed lines (leakages) representing the auxiliary  Block-Marschak polynomials.  Each colored path represents a specific set of types in FRUM. }\label{fig:hasse_RP}
\end{figure}
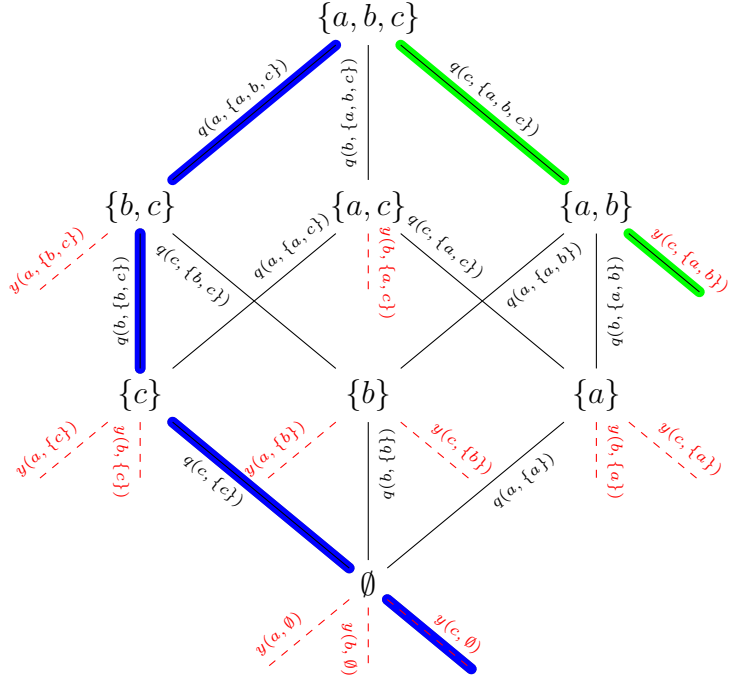

Figure \ref{fig:hasse_RP} generalizes the classical network representation (a directed graph) of partial order sets for our purposes. Each node represents a subset of the grand set of alternatives. Each black solid line indicates a subset relationship between subsets. The Block-Marschak polynomials can be considered as the amount of flow on each black line. In the original network of this Hasse diagram, the outdegree of each node (i.e., a subset)  are equal to the number of alternatives it contains,  and the inflow and outflow of black lines are always equal for each node for classical probabilistic choice $\pi$. Thus, we have  $$\underbrace{ \sum_{b \notin F } q_{\pi}(b, F \cup {b})}_{\text{Inflow}} = \underbrace{\sum_{a \in F} q_{\pi} (a, F)}_{\text{Outflow}} $$ In RUM, each preference ranking corresponds to a path starting from the grand set and ending at the empty set. For example, 
$\{a,b,c\} \rightarrow \{b,c\} \rightarrow \{c\} \rightarrow \emptyset$ (with path $q_\pi(a,\{a,b,c\})  \rightarrow q_\pi(b,\{b,c\})\rightarrow q_\pi(c,\{c\}) $) represents the preference type $ u(a) > u(b) > u(c) $.\footnote{One can refer to \citet{Fiorini2004} for a network analysis of RUM.} However, this attribution of types does not hold in FRUM. In fact, even the above equality no longer holds in FRUM. We discuss below how to recover a similar equality and provide a similar visual representation for types in FRUM.

Compared to RUM, FRUM has two sets of BM polynomials: One for framed alternatives, $q$, and one for non-framed alternatives, $y$. To represent the auxiliary BM polynomials $y$, we introduce new flows, which are represented by dashed red lines. These are always ``leakages'' from the network: $y_{\rho}(a,F)$ denotes the leakage from $F$ for the alternative $a \notin F$. For example, for the node $\{b,c\}$, there is an inflow from $\{a,b,c\}$ but $a \notin \{b,c\}$. Hence, the leakage corresponds exactly to the alternative $a \notin \{b,c\}$. In this way, the outdegree of each node are the number of alternatives in the grand set. Interestingly, if we also take into account all leakages, we recover the equality of inflow and total outflow (outflow plus leakages).\footnote{This result is stated as Lemma~\ref{Thm: Falmagne 3 general} in the Appendix, which is a generalization of \citet{Falmagne1978}'s Theorem 3. We believe this lemma could be of independent interest since it is model-free. We provide a proof of it in the Appendix.} That is,

$$ \underbrace{\sum_{b \notin F } q_\rho(b, F \cup {b})}_{\text{Inflow}} = \underbrace{\sum_{a \in F} q_\rho(a, F)}_{\text{Outflow}} + \underbrace{\sum_{a \notin F} y_\rho(a, F)}_{\text{Leakages}} $$

Given this equality, each type in FRUM can be represented by a path in the new Hasse diagram. Each type corresponds to a path starting from the grand set and ending with a leakage. For example, $q(a,\{a,b,c\}) \rightarrow q(b,\{b,c\}) \rightarrow q(c,\{c\}) \rightarrow y(c,\emptyset)$ (the blue path in Figure \ref{fig:hasse_RP}) represents the choice type who chooses $a$ whenever $a$ is framed, chooses $b$ whenever $b$ is framed but $a$ is not, and chooses $c$ otherwise. In addition, $q(c,\{a,b,c\}) \rightarrow y(c,\{a,b\})$ (the green path in Figure \ref{fig:hasse_RP}) represents the choice type who chooses $c$ independent of frames.

Unlike RUM, each path represents multiple preferences, but they generate the same choice function. For example, $(u_1+v_1)(a)>(u_1+v_1)(b)>(u_1+v_1)(c)> u_1(c)> u_1(b)>u_1(a)$ is represented by the blue line described above ($q(a,\{a,b,c\}) \rightarrow q(b,\{b,c\}) \rightarrow  q(c,\{c\})\rightarrow y(c,\emptyset)$). Similarly, this path also represents the type $(u_2,v_2)$, where  $(u_2+v_2)(a)>(u_2+v_2)(b)>(u_2+v_2)(c)> u_2(c)> u_2(a)>u_2(b)$.  There are several similarities between these types: i) $u_1+v_1$ and $u_2+v_2$ induce the same ordering, ii) $c$ is the worst for $u_1+v_1$ and $u_2+v_2$, and iii) the lower counter sets of $c$ with respect to $u_1$ and $u_2$ are the same. Notice that $c_{(u_1,v_1)}$ and $c_{(u_2,v_2)}$ are inherently indistinguishable even in the deterministic model. Hence, they do not matter for predicting the actual choice. Our identification result is based on choice types rather than utility types. We will return to this in \Cref{Section:FRUM identification}. In general, all utility types represented by the same path induce the same choice behavior in this model. Consider a path $q(a_1,A_0) \rightarrow  \cdots \rightarrow   q(a_{k}, A_{k-1}) \rightarrow y(b,A_k)$ where $b\in X\setminus A_k$, $A_0=X$, and $A_i=A_{i-1} \setminus \{a_i\}$ for $i=1,..., k$. This path captures the choice type $c \in \mathcal{C}_{\text{FUM}}$ where $c(A_k)=b$ and $c(A_{i-1})=a_{i}$ for all $i=1,...,k$.\footnote{While the description of the choice function is incomplete here, it points to a unique $c$ in $\mathcal{C}_{\text{FUM}}$. One can see this by applying the axiom \nameref{Axiom:IIFA}. Moreover, the utility types of each choice function have been discussed in \Cref{{Thm: PFdeterministic uniqueness}}.}

\subsection{Behavioral Characterization}

We now introduce behavioral postulates for FRUM. The first axiom is the non-negativity of BM polynomials in our framework. This non-negativity is closely related to the non-negativity of BM polynomials in the standard domain. In the standard domain, the non-negativity of the BM polynomials means that items must be chosen marginally more if there are fewer available items. RUM requires the non-negativity of the BM polynomials, which is stated below.

\begin{axxm}[Non-negativity of BM]\label{Axiom: non-negativity of BM}
 For  $a \in F$,  $q(a,F) \geq 0$.
\end{axxm}

In RUM, this assumption is both necessary and sufficient. However, it is no longer sufficient in this context. Here, not only must framed alternatives be chosen marginally more if there are fewer framed products, but non-framed alternatives must also be chosen marginally more. Therefore, we also require the non-negativity of $y's$.

\begin{axxm}[Non-negativity of Auxiliary BM]\label{Axiom: non-negativity of ABM}
 For  $a \notin F$,   $y(a,F) \geq 0$.
\end{axxm}

We show that the non-negativity of BM  and Auxiliary BM polynomials fully characterize FRUM. As in RUM, we assume that $\mathcal{D}=2^X$ for the next theorem.

\begin{thm}[Characterization]\label{Thm: F-RUM} Let $\mathcal{D}=2^X$. Then, $\rho$ has a FRUM representation if and only if $\rho$ satisfies \nameref{Axiom: non-negativity of BM} and Auxiliary BM.
\end{thm}

This theorem provides a simple test for frame-dependent random utility using only variations in framed sets. To illustrate usefulness of this theorem, we revisit the example in the introduction. Let $X=\{a,b,c\}$.  Remember that product $a$ was selected 70\% of the time when nothing was framed, 45\% when either product $b$ or $c$ was framed, and only 10\% when both $b$ and $c$ were framed. Note that the market share of non-framed $a$ decreases as the number of framed options increases. It seems like a case where utility maximizers are more drawn to the framed options than the unframed option $a$. Hence, it is tempting to conclude that this example can be rationalized by utility maximization with a heterogeneous population. However, this intuition is misleading.  To see this, we first show that in this example, the non-negativity of the Auxiliary BM polynomials is violated.  We calculate  $$y(a,\emptyset)=\rho(a,\emptyset)-\rho(a,\{b\})-\rho(a,\{c\})+\rho(a,\{b,c\})=0.7-0.45-0.45+0.1<0.$$ 
Therefore, by Theorem \ref{Thm: F-RUM}, this data cannot be seen as the aggregate behavior of frame-dependent utility maximizers. Even with limited data, the premise of frame-dependent utility maximization can be falsified.

Furthermore, in the above example, we assume that $X$ consists of three alternatives and apply our theorem directly using $y$'s. One might expect that the theorem is silent for the above data if $X$ contains more alternatives. In that case, $y(a,\emptyset)$ cannot be calculated with the partial data (\emph{i.e.}, in the absence of other $\rho(a,S\setminus a)$ for some $S\subseteq X$). Nevertheless, the refutation still holds even in this case. We introduce new ``interim'' Auxiliary BM polynomials, which can be calculated with the partial data. Formally, for any $a\notin F' \supseteq F $, we let
$$Y(a,F,F'):=\sum_{B:F\subseteq B \subseteq F' } (-1)^{|B\setminus F|}\rho(a,B)$$
Hence, by definition, $y(a,F)=Y(a,F,X \setminus \{a\})$.  It is routine to show that the non-negativity of $y$'s is equivalent to the non-negativity of $Y$'s. To illustrate this, assume $X=\{a,b,c,d\}$. Then  $y(a,\emptyset)=Y(a,\emptyset,\{b,c\})-Y(a,\{d\},\{b,c,d\})=Y(a,\emptyset,\{b,c\})-y(a,\{d\}) \geq 0$. Hence, $Y(a,\emptyset,\{b,c\}) \geq y(a,\{d\}) \geq 0$. Then a negative $Y$ for some $a\notin F' \supseteq F $ falsifies the model.  Even if $X$ contains more alternatives, we can still deduce that
$$Y(a,\emptyset,\{b,c\})=\rho(a,\emptyset)-\rho(a,\{b\})-\rho(a,\{c\})+\rho(a,\{b,c\})<0$$
which again falsifies the hypothesis that the data is generated by frame-dependent utility maximization.

These are not the only ways to falsify the model. Similarly, it can be violated by the non-negativity of BM polynomials. To see this, we use the example above and let $X=\{a,b,c\}$. Assume product $a$ was selected 80\% of the time when only $a$ was framed, 65\% when either product $b$ or $c$ was framed in addition to $a$, and only 40\% when all of the products were framed. Note that the market share of framed $a$ increases as the number of framed options decreases.
Nevertheless, these data are not consistent with the frame-dependent utility maximization. To see that $$q(a,\{a\})=\rho(a,\{a\})-\rho(a,\{a,b\})-\rho(a,\{a,c\})+\rho(a,\{a,b,c\})=0.8-0.65-0.65+0.4<0.$$ Hence, the frame-dependent utility maximization paradigm requires the non-negativity of BM and Auxiliary BM.\footnote{We would like to highlight that the non-negativity of BM in our setup differs from the non-negativity of BM in RUM. To see this, assume a decision maker with $\pi(a,\{a,b\})=\pi(a,\{a,c\})=0.65$ and $\pi(a,\{a,b,c\})=0.40$. Note that these probabilities are the same as in the above example. Nevertheless, the non-negativity of BM in RUM is not violated.  
$q_{\pi}(a,\{a\})=1-\pi(a,\{a,b\})-\pi(a,\{a,c\})+\pi(a,\{a,b,c\})=1-0.65-0.65+0.4>0.$   Hence, these data could be part of the utility maximization paradigm.} 
One can analogously define a similar ``interim'' term for the (regular) BM polynomials.

Finally, note that the sufficiency proof of \Cref{Thm: F-RUM} is constructive. We provide an algorithm to compute a full distribution of types in the FRUM representation. The algorithm is helpful in applications if one would like to have an estimate of the type distribution. However, just as in the RUM in the standard choice domain, in general, there is no unique distribution over types that can explain the choice data. We explore this further in the next subsection.

\subsection{Uniqueness and Identification for FRUM} \label{Section:FRUM identification}

In this section, we elaborate on to what extent the type distribution can be identified. We next illustrate how one can identify the preference distribution with a simple parametric example involving two products $a$ and $b$. Assume the following data.

\begin{table}[h!]
\centering
\begin{tabular}{r|cccc} 
\toprule
$\ChoiProb_{(\lambda,\gamma)}(\cdot,F)$ & $\{a,b\}$ & $\{a\}$ & $\{b\}$ & $\emptyset$\\
\hline
$a$ & $\lambda$ & $0.7$ & $0.1$ & $\gamma$ \\
$b$ &  $1-\lambda$ & $0.3$ & $0.9$ & $1-\gamma$ \\
\bottomrule \end{tabular}\caption{A Parametric Example} \label{Table:Paremetric Example} \end{table}
To answer whether this choice data has a FRUM representation, we first calculate the corresponding $q_\rho$ and $y_\rho$, which are depicted in the following Hasse diagram in Figure \ref{fig:hasse}. If $\lambda$ and $\gamma$ are between $0.1$ and $0.7$, all $q_\rho$ and $y_\rho$ are non-negative. By Theorem \ref{Thm: F-RUM}, this data has a FRUM representation as long as $\lambda,\gamma \in [0.1,0.7]$. While a characterization theorem can inform us of the existence or non-existence of a particular representation, such a theorem is not necessarily useful for finding such a representation. We now illustrate how one recover the distribution of choice types.

\begin{figure}[h!]
    \begin{center}
    \begin{tikzpicture}[ >=latex,glow/.style={%
    preaction={draw,line cap=round,line join=round,
    opacity=0.3,line width=4pt,#1}},glow/.default=yellow,
    transparency group]
            \tikzstyle{every node} = [rectangle]
        \node (s2) at (0,1.5) {$\emptyset$};
         \node (s2-1) at (-2,0) { };
        \node (s2-3) at (2,0) { };
        \node (s23) at (-2.5,3.5) {$\{b\}$};
        \node (s23-1) at (-4.75,1.75) {};
        \node (s12) at (2.5,3.5) {$\{a\}$};
        \node (s12-3) at (4.75,1.75) {};
        \node (s123) at (0,5.5) {$\{a,b\}$};
        

                \foreach \i/\j/\txt/\p in {
      s12/s12-3/\text{$y(b,\{a\})=.3$}/above,
      s23/s23-1/\text{$y(a,\{b\})=.1$}/above}
                  \draw [red,dashed] (\i) -- node[pos=0.6,sloped,font=\tiny,\p] {\txt} (\j);
            
                \foreach \i/\j/\txt/\p in {
      s2/s2-1/\text{$y(a,\emptyset)=\gamma - 0.1$}/above,
      s2/s2-3/\text{$y(b,\emptyset)=0.7-\gamma$}/above}                \draw [red,dashed] (\i) --              node[pos=0.7,sloped,font=\tiny,\p] {\txt} (\j);
    
         
            
             \foreach \i/\j/\txt/\p in {
      s2/s12/\text{$q(a,\{a\})=0.7-\lambda$}/below,
      s2/s23/\text{$q(b,\{b\})=\lambda-.1$}/below}
            \draw [-] (\i) -- node[pos=0.6,sloped,font=\tiny,\p] {\txt} (\j);
         
          \foreach \i/\j/\txt/\p in {
      s12/s123/\text{$q(b,\{a,b\})=1-\lambda$}/above,
      s23/s123/\text{$q(a,\{a,b\})=\lambda$}/above}
            \draw [-] (\i) -- node[pos=0.5,sloped,font=\tiny,\p] {\txt} (\j);





           \end{tikzpicture}    \end{center}
\caption{Hasse Diagram of $\rho_{(\lambda,\gamma)}$}\label{fig:hasse}
\end{figure}
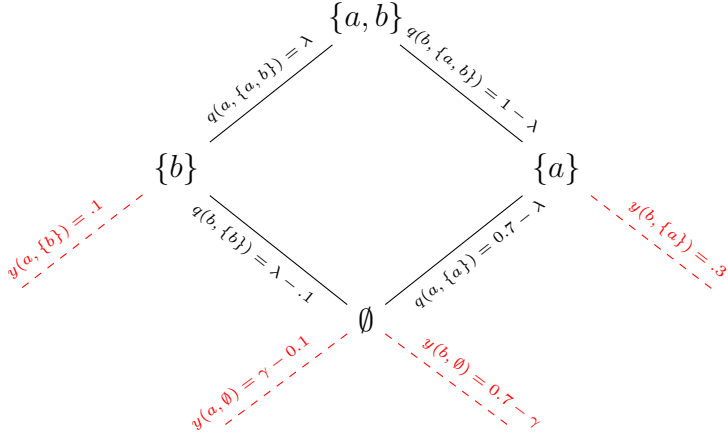

Since our data has a FRUM representation for all $\lambda,\gamma \in [0.1,0.7]$, there exists $\mu$ representing $\ChoiProb_{(\lambda,\gamma)}$. For illustration purposes, take $\lambda=0.4$ and $\gamma=0.6$. We first identify the possible paths and their corresponding weights. Note that each path corresponds to a unique choice type.


The Hasse diagram shows that we must assign at least a flow of $0.1$ to the path $q(a,\{a,b\}) \rightarrow y(a,\{b\})$ and a flow of $0.3$ to the path $q(b,\{a,b\}) \rightarrow y(b,\{a\})$. This means that we identify two types: $c_1$ and $c_2$ always choose $a$ and $b$, respectively, independently of the frame. Hence, $\mu(\{c_1\})=0.1$ and $\mu(\{c_2\})=0.3$. Note that the path $q(a,\{a,b\}) \rightarrow q(b,\{b\})$ carries a flow of $0.3$, and the path $y(b,\emptyset)$ can carry at most a flow of $0.1$. Therefore, we must assign at least a flow of $0.2$ to the path $q(a,\{a,b\}) \rightarrow q(b,\{b\}) \rightarrow y(a,\emptyset)$.  

This path corresponds to the type $c_3$, where $c_{3}(\{a,b\})=c_{3}(\{a\})=c_{3}(\emptyset)=a$ but $c_{3}(\{b\})=b$. We let $\mu(\{c_3\})=0.2$ (there is a degree of freedom here). Since the path $y(a,\emptyset)$ must carry a total of $0.5$, the path $q(b,\{a,b\}) \rightarrow q(a,\{a\}) \rightarrow y(a,\emptyset)$ must carry the residual flow of $0.3$. This path corresponds to $c_4$, where $c_{4}(\{a,b\})=c_{4}(\{b\})=b$ but $c_{4}(\{a\})=c_{4}(\emptyset)=a$. Hence, we must set $\mu(\{c_4\})=0.3$.  

Moreover, since we exhausted the path $q(a,\{a\})$, the rest must flow from the path $q(a,\{a,b\}) \rightarrow q(b,\{b\}) \rightarrow y(b,\emptyset)$. This corresponds to $c_5$, where $c_{5}(\{a,b\})=c_{5}(\{a\})=a$ but $c_{5}(\{b\})=c_{5}(\emptyset)=b$. Hence, we must set $\mu(\{c_5\})=0.1$.  Finally, since we exhausted all the flow, the last choice type, $c_6$, must be assigned a weight of zero. This choice type corresponds to the path $q(b,\{a,b\}) \rightarrow q(a,\{a\}) \rightarrow y(b,\emptyset)$.

We should note that $\mu$ is not unique. Indeed, if we set $\mu'(\{c_3\})=0.25$ or $\mu''(\{c_3\})=0.3$ (the highest possible flow), then we generate two additional representations: $\mu'$ and $\mu''$, respectively. The next table summarizes all three representations. Notice that there is essentially one degree of freedom in the type distribution, which can be captured by $\mu(\{c_3\}) \in [0.2,0.3]$.

\begin{table}[h!]
\begin{center}
\begin{tabular}{l|c|cccc|ccc}
\toprule
\multirow{2}{*}{\text{Paths}} & 
\multirow{2}{*}{\text{Types}}& 
\multicolumn{4}{c|}{$F$} & 
\multirow{2}{*}{$\mu$} & 
\multirow{2}{*}{$\mu'$} & 
\multirow{2}{*}{$\mu''$} \\
\cline{3-6}
&  & $\{a,b\}$ & \ $\{a\}$ \ & \ $\{b\}$ \ & \ $\emptyset$ \ \ \  & & & \\
\midrule
$q(a,\{a,b\}) \rightarrow y(a,\{b\})$ & $c_1$ & $a$ & $a$ & $a$ & $a$ & $0.1$ & $0.1$ & $0.1$ \\
$q(b,\{a,b\}) \rightarrow y(b,\{a\})$ & $c_2$ & $b$ & $b$ & $b$ & $b$ & $0.3$ & $0.3$ & $0.3$ \\
$q(a,\{a,b\}) \rightarrow q(b,\{b\}) \rightarrow y(a,\emptyset)$ & $c_3$ & $a$ & $a$ & $b$ & $a$ & $0.2$ & $0.25$ & $0.3$  \\
$q(b,\{a,b\}) \rightarrow q(a,\{a\}) \rightarrow y(a,\emptyset)$ & $c_4$ & $b$ & $a$ & $b$ & $a$ & $0.3$ & $0.25$ & $0.2$  \\
$q(a,\{a,b\}) \rightarrow q(b,\{b\}) \rightarrow y(b,\emptyset)$ & $c_5$ & $a$ & $a$ & $b$ & $b$ & $0.1$ & $0.05$ & $0$  \\
$q(b,\{a,b\}) \rightarrow q(a,\{a\}) \rightarrow y(b,\emptyset)$ & $c_6$ & $b$ & $a$ & $b$ & $b$ & $0$ & $0.05$ & $0.1$  \\
\bottomrule
\end{tabular}
\end{center}
\caption{Three different FRUM representations for $\ChoiProb_{(0.4,0.6)}$}
\label{Table:Branch-independent}
\end{table}

We now generalize the above discussion to state our uniqueness result in the following proposition.


\begin{prop}[Uniqueness and identification of FRUM]\label{Thm: FRUM uniqueness}
Let $\mu^1$ and $\mu^2$ be two FRUM representations of the same choice data. Then for every $F\subseteq X$, $b\notin  F$ and $i=1,2$, \\
1) $y_\rho(b,F)=\mu^i\big(\{c \in \mathcal{C}_{\text{FUM}} | \ 
       b=c(F) \  \text{ and }  \   x = c(\{x\}) \ \text{ for all } x \notin F  \} \big)$,\\
2) $  q_\rho(b,F \cup \{b\}) = \mu^i \big( \{c \in \mathcal{C}_{\text{FUM}} | \ 
       b=c(F\cup \{b\}) \  \text{ and }  \   x = c(\{x,b\}) \ \text{ for all } x \notin F \}  \big)$.
\end{prop}

Notice that the uniqueness result of the FRUM is closely related to the uniqueness result of the standard RUM model. We first focus on $y$. The first condition says that the choice types captured by $y_\rho(b,F)$ must choose $b \notin F$ when $F$ is framed. Note that since these choice types have a FUM representation, we know that it captures all the utility types where $u(b)$ is greater than $u(x)$ for all $x \in X \setminus \{b\}$ and $(u+v)(z)$ for all $z \in F$. Moreover, the choice types captured by $y_\rho(b,F)$ also choose $x$ whenever $x$ is framed. Therefore, these types also assign a higher $u+v$ value to all alternatives in $X \setminus F$ than $u(b)$.

The interpretation for $q_\rho(b,F\cup \{b\})$ is analogous. First, choice types captured by $q_\rho(b,F\cup \{b\})$ must  choose $b$ when $F \cup \{b\}$ is framed, representing all the utility types where $(u+v)(b)$ is greater than $u(x)$ for all $x \in X$ and $(u+v)(z)$ for all $z \in F$. Also, they choose $x \notin F$ whenever $\{x,b\}$ is framed, so the underlying utility types also assign a higher $u+v$ value to all alternatives in $X \setminus (F \cup \{x\})$ than $(u+v)(b)$.

As we mentioned before, in RUM, it has been known that the distribution over choice types cannot, in general, be uniquely recovered from stochastic choice data (\citealp{Falmagne1978,Fishburn1998}). Recently emerging literature provides remedies for the non-uniqueness of RUM (\citealp{luce1959individual,Gul_Pesendorfer2006,turansick2022,Honda2021,cattaneo2024attentionoverload,YANG2023105650,SULEYMANOV2024,Filiz-Ozbay_Masatlioglu_PRC2023,caradonna2024}). 

In RUM framework, \citet{turansick2022} first provides necessary and sufficient conditions on the Block–Marschak polynomials, which guarantee that a random choice rule has a unique RUM representation. Since our model is distinct from RUM, those conditions are not useful in our setup. Having said that \citet{caradonna2024} recently generalizes the result of \citet{turansick2022} which covers our setup. Hence, one can apply their result in our framework to provide the necessary and sufficient conditions for uniqueness in FRUM.

Since \citet{SULEYMANOV2024} achieves uniqueness without sacrificing the explanatory power of RUM, we adopt his approach to attain unique recovery for FRUM. \citet{SULEYMANOV2024} considers a particular representation of RUM, namely branch-independent RUM, that enjoys a uniqueness property in identification while maintaining the same explanatory power as RUM. The core idea is to apply the construction of \citet{Falmagne1978} in recovering the preference types. Note that our proof indeed provides a specific construction of preference types that parallels that of \citet{Falmagne1978}. Therefore, one can apply a similar idea in our setting. However, there is an important distinction. Recall that in RUM, each path in the Hesse Diagram corresponds to a unique preference type, where \citet{Falmagne1978}'s construction uniquely assigns each preference a weight according to the path. Nevertheless, in FRUM, each path can correspond to multiple utility types. Therefore, we only consider the choice type each path represents. Consider a path $q(a_1,A_0) \rightarrow \cdots \rightarrow q(a_{k}, A_{k-1}) \rightarrow y(b,A_k)$ where $b \in X \setminus A_k$, $A_0 = X$, and $A_i = A_{i-1} \setminus \{a_i\}$ for $i = 1, \dots, k$. It captures the choice type $c \in \mathcal{C}_{\text{FUM}}$ where $c(A_k) = b$ and $c(A_{i-1}) = a_{i}$ for all $i = 1, \dots, k$. The corresponding weight is
\begin{align*}
\prod_{i=1}^k \frac{q(a_i, A_{i-1})}{q(A_{i-1},A_{i-1})+y(X\setminus A_{i-1},A_{i-1})} \cdot \frac{y(b, A_k)}{q(A_k,A_k)+y(X\setminus A_k,A_k)}, 
\end{align*} where $q(S,R)=\sum_{x\in S}q(x,R)$ if $S\neq \emptyset$ and $q(\emptyset,R)=0$.\footnote{In the proof, we provide its recursive construction, which is equivalent to this formula.} An analogous definition is applied to $y$. As an illustration, we revisit the example in \Cref{Table:Paremetric Example}. If we apply the formula to the path $q(a,\{a,b\}) \rightarrow q(b,\{b\}) \rightarrow y(a,\emptyset)$, we have $\frac{q(a,\{a,b\})}{q(\{a,b\},\{a,b\})}\frac{q(b,\{b\})}{q(b,\{b\}) + y(a,\{b\})} \frac{y(a,\emptyset)}{y(a,\emptyset)+y(b,\emptyset)} = 0.25$. This corresponds to the choice type $c(\{a,b\}) = c(\{a\}) = c(\emptyset) = a$ and $c(\{b\}) = b$, where the choice $c(\{a\})$ is implied by $c(\{a,b\})$ since $c \in \mathcal{C}_{\text{FUM}}$. Note that this choice type is $c_3$ as shown in \Cref{Table:Branch-independent}, and the weight equals exactly to $\mu'(\{c_3\})$. Therefore, $\mu'$ is the unique branch-independent representation of the choice data.

\section{Parametric Probabilistic Choice}\label{sec:parametric}

We now study the idea of frame-dependent utility with probabilistic data based on the parametric Luce model. All the information about the models' predictions can be summarized by a finite set of parameters depending only on the alternatives. As is typically the case with parametric models, our parametric models offer three advantages over the non-parametric versions we study in the next section. First, our parametric model is ready to use in applications because of their analytic and computational tractability. Second, the model possesses strong uniqueness properties. Third, it has sharp identification results for designing policies. 

We make the following positivity assumption throughout this section to foster a comparison between the two channels in the parametric models. Notice that this positivity property, as argued by \citet{McFadden1973ConditionalBehavior}, cannot be refuted based on any finite data set.

\begin{Assumption}
For every $x$ and $F$, $\rho(x,F)>0$.
\end{Assumption}

Under frame-dependent utility, the valuation of the framed item is calculated by $u+v$. In terms of the Luce model, we assume that for each available alternative $x$, there are two alternative-specific parameters $u(x)>0$ and $v(x)\geq0$. Like the Luce model, $u$ represents the crude measure of utility value without any framing. On the other hand, $v(x)$ captures the boost when alternative $x$ is framed. We sometimes use $(u,v)$ as the primitive of the model for convenience. For notational simplicity, we write $f(A)$ as shorthand for $\sum_{x\in A} f(x)$ for any function $f$. The DM makes choices based on the utility value. Therefore, the choice probability of a framed alternative equals its weight $(u+v)(x)$ divided by the total weight in the choice set. The choice probability of a non-framed alternative is equal to its weight $u(x)$ divided by the total weight in the choice set. We summarize this into the following definition.

\begin{defn}
A choice rule $\rho$ has a F-Luce  representation  if there exist  functions $u:X \rightarrow \mathbbm{R}_{++}$ and  $v:X \rightarrow \mathbbm{R}_{+}$ such that for $x \in X$, \begin{equation*}
\rho(x,F)= 
\begin{dcases}
    \frac{u(x)+v(x)}{u(X)+v(F)}  &\text{ if } x \in F\\
    \frac{u(x)}{u(X)+v(F)}   &\text{ otherwise}\\    
\end{dcases} \tag{F-Luce}  \end{equation*} for all $F \in \mathcal{D}$.
\end{defn}

\Cref{fig:F-Luce} illustrates both the similarity and the added richness introduced when moving from the classical Luce model to the F-Luce model. In the left panel, which corresponds to the standard Luce model, there are three alternatives and only four non-trivial choice problems: the grand set and three binary choice problems. Here, we slightly abuse notation. For example, \(\pi(xy)\) captures the probability of choosing between \(x\) and \(y\) only, assigning zero probability to \(z\) since \(z\) is not available. In this setting, the structure is relatively simple, and the model essentially reports how the ratios of choice probabilities among the available alternatives remain independent of the presence or absence of any other non-available alternative. This is reflected by $\pi(xy)$ being a projection of $\pi(X)$ onto the line $xy$.

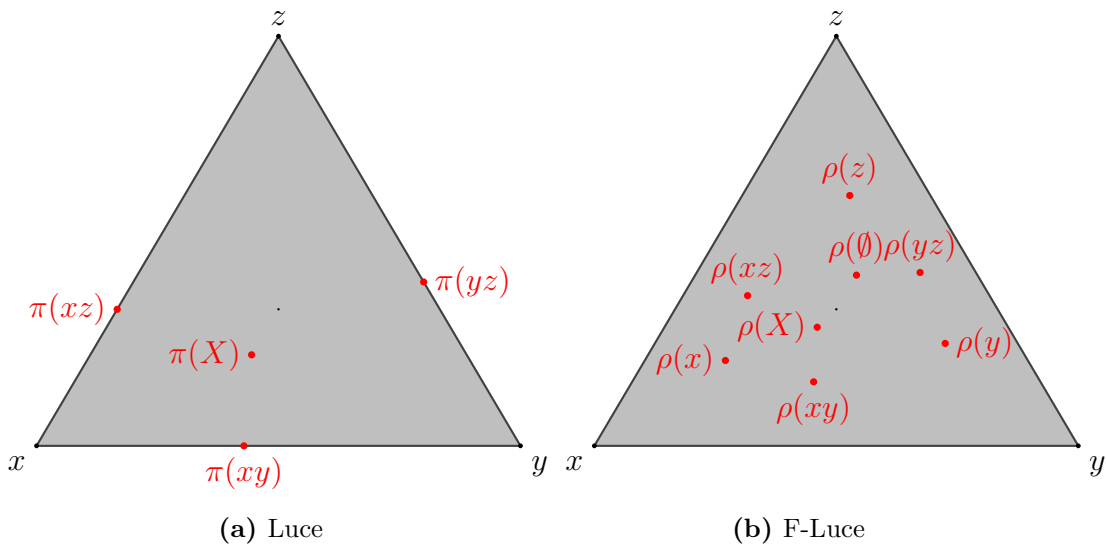
\begin{figure}[h!]
 \centering
      \begin{subfigure}[t]{0.45\textwidth}
        \centering
   \begin{tikzpicture}[line cap=round,line join=round,                    3d view={135}{45},                    scale=4.5]


\def\ux{3}   
\def\uy{2}   
\def\uz{1}   

\def\vx{1}   
\def\vy{1}   
\def\vz{1}   

\def\xx{1}   
\def\yy{1}   
\def\zz{1}   

\pgfmathsetmacro\pxxy{(\ux+\vx)/(\ux+\vx+\uy+\vy)} 
\pgfmathsetmacro\pyxy{(\uy+\vy)/(\ux+\vx+\uy+\vy)} %

\pgfmathsetmacro\pxxz{(\ux+\vx)/(\ux+\vx+\vz+\uz)} 
\pgfmathsetmacro\pyxz{(0)/(\ux+\vx+\vz+\uz)} %

\pgfmathsetmacro\pxyz{(0)/(\vy+\uy+\vz+\uz)} 
\pgfmathsetmacro\pyyz{(\uy+\vy)/(\vy+\uy+\vz+\uz)} %

\pgfmathsetmacro\pxX{(\ux+\vx)/(\ux+\vx+\vy+\uy+\vz+\uz)} 
\pgfmathsetmacro\pyX{(\uy+\vy)/(\ux+\vx+\vy+\uy+\vz+\uz)} %

\pgfmathsetmacro\pzxy{\zz*(1-\pxxy/\xx-\pyxy/\yy)}
\pgfmathsetmacro\pzxz{\zz*(1-\pxxz/\xx-\pyxz/\yy)}
\pgfmathsetmacro\pzyz{\zz*(1-\pxyz/\xx-\pyyz/\yy)}
\pgfmathsetmacro\pzX{\zz*(1-\pxX/\xx-\pyX/\yy)}

\coordinate (O)  at (0,0,0);
\coordinate (CE)  at (1/3,1/3,1/3);
\coordinate (A)  at (\xx,0,0);
\coordinate (B)  at (0,\yy,0);
\coordinate (C)  at (0,0,\zz);
\coordinate (Pxy)  at (\pxxy,\pyxy,\pzxy);
\coordinate (Pxz)  at (\pxxz,\pyxz,\pzxz);
\coordinate (Pyz)  at (\pxyz,\pyyz,\pzyz);
\coordinate (PX)  at (\pxX,\pyX,\pzX);


\draw[gray, dotted, very thin] (1,0,0) -- (Pyz);

\draw[gray, dotted, very thin] (0,1,0) -- (Pxz);

\draw[gray, dotted, very thin] (0,0,1) -- (Pxy);
\draw[gray!50!black,thick,fill=gray!50,fill opacity=0.1] (A) -- (B) -- (C) -- cycle;
\fill (A)  circle (.2pt) node[below left] {$x$};
\fill (B)  circle (.2pt) node[below right] {$y$};
\fill (C)  circle (.2pt) node[above] {$z$};

\fill[red] (Pxy)  circle (.3pt) node[below] {$\pi(xy)$};
\fill[red] (Pyz)  circle (.3pt) node[right] {$\pi(yz)$};
\fill[red] (Pxz)  circle (.3pt) node[left] {$\pi(xz)$};
\fill[red] (PX)  circle (.3pt) node[left] {$\pi(X)$};
\fill (CE)  circle (.1pt) node[below left] { };
\end{tikzpicture}   

 \caption{Luce}
    \end{subfigure}
   \begin{subfigure}[t]{0.4\textwidth}
        \centering
\begin{tikzpicture}[line cap=round,line join=round,                    3d view={135}{45},                    scale=4.5]


\def\ux{1.5}   
\def\uy{2}   
\def\uz{2.5}   

\def\vx{6}   
\def\vy{4}   
\def\vz{3}   
\def\xx{1}   
\def\yy{1}   
\def\zz{1}   

\pgfmathsetmacro\pxe{\ux/(\ux+\uy+\uz)} 
\pgfmathsetmacro\pye{\uy/(\ux+\uy+\uz)}

\pgfmathsetmacro\pxx{(\ux+\vx)/(\ux+\vx+\uy+\uz)} 
\pgfmathsetmacro\pyx{(\uy)/(\ux+\vx+\uy+\uz)}%

\pgfmathsetmacro\pxy{(\ux)/(\ux+\uy+\vy+\uz)} 
\pgfmathsetmacro\pyy{(\uy+\vy)/(\ux+\uy+\vy+\uz)}%

\pgfmathsetmacro\pxz{(\ux)/(\ux+\uy+\vz+\uz)} 
\pgfmathsetmacro\pyz{(\uy)/(\ux+\uy+\vz+\uz)} %

\pgfmathsetmacro\pxxy{(\ux+\vx)/(\ux+\vx+\uy+\vy+\uz)} 
\pgfmathsetmacro\pyxy{(\uy+\vy)/(\ux+\vx+\uy+\vy+\uz)} %

\pgfmathsetmacro\pxxz{(\ux+\vx)/(\ux+\vx+\uy+\vz+\uz)} 
\pgfmathsetmacro\pyxz{(\uy)/(\ux+\vx+\uy+\vz+\uz)} %

\pgfmathsetmacro\pxyz{(\ux)/(\ux+\vy+\uy+\vz+\uz)} 
\pgfmathsetmacro\pyyz{(\uy+\vy)/(\ux+\vy+\uy+\vz+\uz)} %

\pgfmathsetmacro\pxX{(\ux+\vx)/(\ux+\vx+\vy+\uy+\vz+\uz)} 
\pgfmathsetmacro\pyX{(\uy+\vy)/(\ux+\vx+\vy+\uy+\vz+\uz)} %

\pgfmathsetmacro\pze{\zz*(1-\pxe/\xx-\pye/\yy)}
\pgfmathsetmacro\pzx{\zz*(1-\pxx/\xx-\pyx/\yy)}
\pgfmathsetmacro\pzy{\zz*(1-\pxy/\xx-\pyy/\yy)}
\pgfmathsetmacro\pzz{\zz*(1-\pxz/\xx-\pyz/\yy)}
\pgfmathsetmacro\pzxy{\zz*(1-\pxxy/\xx-\pyxy/\yy)}
\pgfmathsetmacro\pzxz{\zz*(1-\pxxz/\xx-\pyxz/\yy)}
\pgfmathsetmacro\pzyz{\zz*(1-\pxyz/\xx-\pyyz/\yy)}
\pgfmathsetmacro\pzX{\zz*(1-\pxX/\xx-\pyX/\yy)}

\coordinate (O)  at (0,0,0);
\coordinate (CE)  at (1/3,1/3,1/3);
\coordinate (A)  at (\xx,0,0);
\coordinate (B)  at (0,\yy,0);
\coordinate (C)  at (0,0,\zz);
\coordinate (Pe)  at (\pxe,\pye,\pze);
\coordinate (Px)  at (\pxx,\pyx,\pzx);
\coordinate (Py)  at (\pxy,\pyy,\pzy);
\coordinate (Pz)  at (\pxz,\pyz,\pzz);
\coordinate (Pxy)  at (\pxxy,\pyxy,\pzxy);
\coordinate (Pxz)  at (\pxxz,\pyxz,\pzxz);
\coordinate (Pyz)  at (\pxyz,\pyyz,\pzyz);
\coordinate (PX)  at (\pxX,\pyX,\pzX);

\fill[white] (.5, .5,0)  circle (.0005pt) node[below] {$\rho(x)$};

\draw[gray, dotted, very thin] (1,0,0) -- (Pe);
\draw[gray, dotted, very thin] (1,0,0) -- (Py);
\draw[gray, dotted, very thin] (1,0,0) -- (Pz);
\draw[gray, dotted, very thin] (1,0,0) -- (Pyz);
\draw[gray, dotted, very thin] (0,1,0) -- (Pe);
\draw[gray, dotted, very thin] (0,1,0) -- (Px);
\draw[gray, dotted, very thin] (0,1,0) -- (Pz);
\draw[gray, dotted, very thin] (0,1,0) -- (Pxz);
\draw[gray, dotted, very thin] (0,0,1) -- (Pe);
\draw[gray, dotted, very thin] (0,0,1) -- (Py);
\draw[gray, dotted, very thin] (0,0,1) -- (Px);
\draw[gray, dotted, very thin] (0,0,1) -- (Pxy);
\draw[gray!50!black,thick,fill=gray!50,fill opacity=0.1] (A) -- (B) -- (C) -- cycle;
\fill (A)  circle (.2pt) node[below left] {$x$};
\fill (B)  circle (.2pt) node[below right] {$y$};
\fill (C)  circle (.2pt) node[above] {$z$};
\fill[red] (Pe)  circle (.3pt) node[above] {$\rho(\emptyset)$};
\fill[red] (Px)  circle (.3pt) node[left] {$\rho(x)$};
\fill[red] (Py)  circle (.3pt) node[right] {$\rho(y)$};
\fill[red] (Pz)  circle (.3pt) node[above] {$\rho(z)$};
\fill[red] (Pxy)  circle (.3pt) node[below] {$\rho(xy)$};
\fill[red] (Pyz)  circle (.3pt) node[above] {$\rho(yz)$};
\fill[red] (Pxz)  circle (.3pt) node[above] {$\rho(xz)$};
\fill[red] (PX)  circle (.3pt) node[left] {$\rho(X)$};
\fill (CE)  circle (.1pt) node[below left] { };
\end{tikzpicture}
 \caption{F-Luce}
    \end{subfigure}%

    \caption{Luce Model versus F-Luce Model
    }
    \label{fig:F-Luce}
\end{figure}

In the right panel of \Cref{fig:F-Luce}, we illustrate an F-Luce model with three alternatives. There are eight non-trivial observations corresponding to each frame. Here, framing an alternative can boost its choice probability relative to others. Unlike the classical Luce model, \(\rho(xy)\) assigns a positive choice probability to \(z\) even though it is not framed, which is placed strictly inside the probability simplex. Despite these differences, certain regularities persist across the two models. For instance, the data for \(\{x,y\}\) and \(X\) lie on the same line, indicating that the relative choice probability between \(x\) and \(y\) remains constant as long as both are framed. Thus, while the F-Luce model introduces an additional layer of complexity through framing effects, it still preserves the fundamental property of the original Luce model. Moreover, the relative choice probability between \(x\) and \(y\) stays the same whether $z$ is framed or not. For example, when neither \(x\) nor \(y\) is framed, the ratio stays constant, indicated by the fact that $\rho(\emptyset)$ and $\rho(z)$ are on the same line. Furthermore, when \(x\) is framed but \(y\) is not, the ratio also stays constant, since $\rho(xz)$ and $\rho(x)$ are on the same line. We will discuss more of the general principles behind F-Luce in the next section.

There are several interesting special cases of the model. All these special cases are one-parameter extensions of the original Luce model capturing the framing effects. These examples illustrate the richness of the F-Luce model.

First, we examine the situation that the impact of framing remains constant across alternatives. Here, while the underlying utility $u$ varies, the framing effect adds a fixed gain $\bar{v}$ to the utility across all alternatives. In the absence of framing, the choice probabilities, $\rho(\cdot,\emptyset)$, directly reflect the underlying preferences among alternatives. Introducing framing, however, alters these probabilities. Framing a specific alternative increases its choice probability, but framing all alternatives uniformly diminishes the distinctions between them. Consequently, the choice probabilities, $\rho(\cdot,X)$, converge towards a convex combination of the uniform distribution and $\rho(\cdot,\emptyset)$. In an extreme scenario where $\bar{v}=0$, framing has no effect ($\rho(x,F)$ is equal to $\rho(x,\emptyset)$ for all $x$ and for all $F$). As $\bar{v}$ increases,  $\rho(\cdot,X)$ approaches the uniform distribution, further reducing the differentiation between alternatives. This scenario captures environments in which framing induces a consistent effect across all products.

Second, we imagine a case where the underlying utility $u$ is constant across alternatives, and the only difference comes from the boost in utility $v$. The model is represented by $(\bar{u},v)$, where $\bar{u}$ is a constant. In this case, there is no distinction between alternatives initially. Hence the distribution of choices is uniform if there is no frame. The framing effect breaks this uniformity in favor of framed alternatives. $\rho(\cdot,X)$ reflects the differential effect of framing across alternatives. This special case captures environments where customers either are indifferent between alternatives or have no information about products. Framing eliminates indifference or provides additional information.  

Finally, we explore scenarios where the framing effect is proportional to the underlying utility. In this setting, alternatives with higher utility receive larger boosts from framing. For example, promoting best-sellers amplifies their sales even further. The model for this case is represented as $(u,\lambda u)$, where $\lambda\geq 0$ quantifies the framing effect. For example, it could be a measure of the effectiveness of recommendations. If $\lambda$ is small, $\rho(x,F)$ is in close proximity to $\rho(x,\emptyset)$ for all $x$ and for all $F$, indicating a minimal impact of framing. Conversely, when $\lambda$ is large, framing significantly increases the choice probabilities of each alternative relative to the no-framing scenario. Interestingly, framing all alternatives simultaneously has no impact on their relative choice probabilities, resulting in $\rho(x,\emptyset) =\rho(x,X)$ for all $x$. All three cases are illustrated in \Cref{fig:special}.

\begin{figure}[h!]
    \begin{tikzpicture}[line cap=round,line join=round,                    3d view={135}{45},                    scale=3]


\def\ux{2}   
\def\uy{3}   
\def\uz{1}   

\def\vx{3}   
\def\vy{3}   
\def\vz{3}   
\def\xx{1}   
\def\yy{1}   
\def\zz{1}   

\pgfmathsetmacro\pxe{\ux/(\ux+\uy+\uz)} 
\pgfmathsetmacro\pye{\uy/(\ux+\uy+\uz)}

\pgfmathsetmacro\pxx{(\ux+\vx)/(\ux+\vx+\uy+\uz)} 
\pgfmathsetmacro\pyx{(\uy)/(\ux+\vx+\uy+\uz)}%

\pgfmathsetmacro\pxy{(\ux)/(\ux+\uy+\vy+\uz)} 
\pgfmathsetmacro\pyy{(\uy+\vy)/(\ux+\uy+\vy+\uz)}%

\pgfmathsetmacro\pxz{(\ux)/(\ux+\uy+\vz+\uz)} 
\pgfmathsetmacro\pyz{(\uy)/(\ux+\uy+\vz+\uz)} %

\pgfmathsetmacro\pxxy{(\ux+\vx)/(\ux+\vx+\uy+\vy+\uz)} 
\pgfmathsetmacro\pyxy{(\uy+\vy)/(\ux+\vx+\uy+\vy+\uz)} %

\pgfmathsetmacro\pxxz{(\ux+\vx)/(\ux+\vx+\uy+\vz+\uz)} 
\pgfmathsetmacro\pyxz{(\uy)/(\ux+\vx+\uy+\vz+\uz)} %

\pgfmathsetmacro\pxyz{(\ux)/(\ux+\vy+\uy+\vz+\uz)} 
\pgfmathsetmacro\pyyz{(\uy+\vy)/(\ux+\vy+\uy+\vz+\uz)} %

\pgfmathsetmacro\pxX{(\ux+\vx)/(\ux+\vx+\vy+\uy+\vz+\uz)} 
\pgfmathsetmacro\pyX{(\uy+\vy)/(\ux+\vx+\vy+\uy+\vz+\uz)} %

\pgfmathsetmacro\pze{\zz*(1-\pxe/\xx-\pye/\yy)}
\pgfmathsetmacro\pzx{\zz*(1-\pxx/\xx-\pyx/\yy)}
\pgfmathsetmacro\pzy{\zz*(1-\pxy/\xx-\pyy/\yy)}
\pgfmathsetmacro\pzz{\zz*(1-\pxz/\xx-\pyz/\yy)}
\pgfmathsetmacro\pzxy{\zz*(1-\pxxy/\xx-\pyxy/\yy)}
\pgfmathsetmacro\pzxz{\zz*(1-\pxxz/\xx-\pyxz/\yy)}
\pgfmathsetmacro\pzyz{\zz*(1-\pxyz/\xx-\pyyz/\yy)}
\pgfmathsetmacro\pzX{\zz*(1-\pxX/\xx-\pyX/\yy)}

\coordinate (O)  at (0,0,0);
\coordinate (A)  at (\xx,0,0);
\coordinate (B)  at (0,\yy,0);
\coordinate (CE)  at (1/3,1/3,1/3);
\coordinate (C)  at (0,0,\zz);
\coordinate (Pe)  at (\pxe,\pye,\pze);
\coordinate (Px)  at (\pxx,\pyx,\pzx);
\coordinate (Py)  at (\pxy,\pyy,\pzy);
\coordinate (Pz)  at (\pxz,\pyz,\pzz);
\coordinate (Pxy)  at (\pxxy,\pyxy,\pzxy);
\coordinate (Pxz)  at (\pxxz,\pyxz,\pzxz);
\coordinate (Pyz)  at (\pxyz,\pyyz,\pzyz);
\coordinate (PX)  at (\pxX,\pyX,\pzX);

\draw[line width=0.01mm, gray, dotted] (1,0,0) -- (Pe);
\draw[gray, dotted, very thin] (1,0,0) -- (Py);
\draw[gray, dotted, very thin] (1,0,0) -- (Pz);
\draw[gray, dotted, very thin] (1,0,0) -- (Pyz);
\draw[gray, dotted, very thin] (0,1,0) -- (Pe);
\draw[gray, dotted, very thin] (0,1,0) -- (Px);
\draw[gray, dotted, very thin] (0,1,0) -- (Pz);
\draw[gray, dotted, very thin] (0,1,0) -- (Pxz);
\draw[gray, dotted, very thin] (0,0,1) -- (Pe);
\draw[gray, dotted, very thin] (0,0,1) -- (Py);
\draw[gray, dotted, very thin] (0,0,1) -- (Px);
\draw[gray, dotted, very thin] (0,0,1) -- (Pxy);
\draw[gray!50!black,thick,fill=gray!50,fill opacity=0.1] (A) -- (B) -- (C) -- cycle;
\fill (A)  circle (.2pt) node[below left] {$x$};
\fill (B)  circle (.2pt) node[below right] {$y$};
\fill (C)  circle (.2pt) node[above] {$z$};
\fill[red] (Pe)  circle (.3pt) node[right] {\tiny{$\rho(\emptyset)$}};
\fill[red] (Px)  circle (.3pt) node[left] {\tiny{$\rho(x)$}};
\fill[red] (Py)  circle (.3pt) node[right] {\tiny{$\rho(y)$}};
\fill[red] (Pz)  circle (.3pt) node[above] {\tiny{$\rho(z)$}};
\fill[red] (Pxy)  circle (.3pt) node[xshift=0cm, yshift=-0.15cm] {\tiny{$\rho(xy)$}};
\fill[red] (Pyz)  circle (.3pt) node[above] {\tiny{$\rho(yz)$}};
\fill[red] (Pxz)  circle (.3pt) node[above] {\tiny{$\rho(xz)$}};
\fill[red] (PX)  circle (.3pt) node[left] {\tiny{$\rho(X)$}};
\fill (CE)  circle (.1pt) node[below left] { };
\node at (CE) {\tiny\(\star\)};
\end{tikzpicture}  \begin{tikzpicture}[line cap=round,line join=round,                    3d view={135}{45},                    scale=3]


\def\ux{1}   
\def\uy{1}   
\def\uz{1}   

\def\vx{3}   
\def\vy{2}   
\def\vz{1}   
\def\xx{1}   
\def\yy{1}   
\def\zz{1}   

\pgfmathsetmacro\pxe{\ux/(\ux+\uy+\uz)} 
\pgfmathsetmacro\pye{\uy/(\ux+\uy+\uz)}

\pgfmathsetmacro\pxx{(\ux+\vx)/(\ux+\vx+\uy+\uz)} 
\pgfmathsetmacro\pyx{(\uy)/(\ux+\vx+\uy+\uz)}%

\pgfmathsetmacro\pxy{(\ux)/(\ux+\uy+\vy+\uz)} 
\pgfmathsetmacro\pyy{(\uy+\vy)/(\ux+\uy+\vy+\uz)}%

\pgfmathsetmacro\pxz{(\ux)/(\ux+\uy+\vz+\uz)} 
\pgfmathsetmacro\pyz{(\uy)/(\ux+\uy+\vz+\uz)} %

\pgfmathsetmacro\pxxy{(\ux+\vx)/(\ux+\vx+\uy+\vy+\uz)} 
\pgfmathsetmacro\pyxy{(\uy+\vy)/(\ux+\vx+\uy+\vy+\uz)} %

\pgfmathsetmacro\pxxz{(\ux+\vx)/(\ux+\vx+\uy+\vz+\uz)} 
\pgfmathsetmacro\pyxz{(\uy)/(\ux+\vx+\uy+\vz+\uz)} %

\pgfmathsetmacro\pxyz{(\ux)/(\ux+\vy+\uy+\vz+\uz)} 
\pgfmathsetmacro\pyyz{(\uy+\vy)/(\ux+\vy+\uy+\vz+\uz)} %

\pgfmathsetmacro\pxX{(\ux+\vx)/(\ux+\vx+\vy+\uy+\vz+\uz)} 
\pgfmathsetmacro\pyX{(\uy+\vy)/(\ux+\vx+\vy+\uy+\vz+\uz)} %

\pgfmathsetmacro\pze{\zz*(1-\pxe/\xx-\pye/\yy)}
\pgfmathsetmacro\pzx{\zz*(1-\pxx/\xx-\pyx/\yy)}
\pgfmathsetmacro\pzy{\zz*(1-\pxy/\xx-\pyy/\yy)}
\pgfmathsetmacro\pzz{\zz*(1-\pxz/\xx-\pyz/\yy)}
\pgfmathsetmacro\pzxy{\zz*(1-\pxxy/\xx-\pyxy/\yy)}
\pgfmathsetmacro\pzxz{\zz*(1-\pxxz/\xx-\pyxz/\yy)}
\pgfmathsetmacro\pzyz{\zz*(1-\pxyz/\xx-\pyyz/\yy)}
\pgfmathsetmacro\pzX{\zz*(1-\pxX/\xx-\pyX/\yy)}

\coordinate (O)  at (0,0,0);
\coordinate (CE)  at (1/3,1/3,1/3);
\coordinate (A)  at (\xx,0,0);
\coordinate (B)  at (0,\yy,0);
\coordinate (C)  at (0,0,\zz);
\coordinate (Pe)  at (\pxe,\pye,\pze);
\coordinate (Px)  at (\pxx,\pyx,\pzx);
\coordinate (Py)  at (\pxy,\pyy,\pzy);
\coordinate (Pz)  at (\pxz,\pyz,\pzz);
\coordinate (Pxy)  at (\pxxy,\pyxy,\pzxy);
\coordinate (Pxz)  at (\pxxz,\pyxz,\pzxz);
\coordinate (Pyz)  at (\pxyz,\pyyz,\pzyz);
\coordinate (PX)  at (\pxX,\pyX,\pzX);

\node at (CE) {\tiny\(\star\)};
\draw[gray, dotted, very thin] (1,0,0) -- (Pe);
\draw[gray, dotted, very thin] (1,0,0) -- (Py);
\draw[gray, dotted, very thin] (1,0,0) -- (Pz);
\draw[gray, dotted, very thin] (1,0,0) -- (Pyz);
\draw[gray, dotted, very thin] (0,1,0) -- (Pe);
\draw[gray, dotted, very thin] (0,1,0) -- (Px);
\draw[gray, dotted, very thin] (0,1,0) -- (Pz);
\draw[gray, dotted, very thin] (0,1,0) -- (Pxz);
\draw[gray, dotted, very thin] (0,0,1) -- (Pe);
\draw[gray, dotted, very thin] (0,0,1) -- (Py);
\draw[gray, dotted, very thin] (0,0,1) -- (Px);
\draw[gray, dotted, very thin] (0,0,1) -- (Pxy);
\draw[gray!50!black,thick,fill=gray!50,fill opacity=0.1] (A) -- (B) -- (C) -- cycle;
\fill (A)  circle (.2pt) node[below left] {$x$};
\fill (B)  circle (.2pt) node[below right] {$y$};
\fill (C)  circle (.2pt) node[above] {$z$};
\fill[red] (Pe)  circle (.3pt) node[above] {\tiny{$\rho(\emptyset)$}};
\fill[red] (Px)  circle (.3pt) node[left] {\tiny{$\rho(x)$}};
\fill[red] (Py)  circle (.3pt) node[right] {\tiny{$\rho(y)$}};
\fill[red] (Pz)  circle (.3pt) node[above] {\tiny{$\rho(z)$}};
\fill[red] (Pxy)  circle (.3pt) node[below] {\tiny{$\rho(xy)$}};
\fill[red] (Pyz)  circle (.3pt) node[above] {\tiny{$\rho(yz)$}};
\fill[red] (Pxz)  circle (.3pt) node[above] {\tiny{$\rho(xz)$}};
\fill[red] (PX)  circle (.3pt) node[left] {\tiny{$\rho(X)$}};
\end{tikzpicture}     \begin{tikzpicture}[line cap=round,line join=round,                    3d view={135}{45},                    scale=3]


\def\ux{2}   
\def\uy{1.7}   
\def\uz{3}   

\def\l{2}

\def\vx{\ux*\l}   
\def\vy{\uy*\l}   
\def\vz{\uz*\l}   
\def\xx{1}   
\def\yy{1}   
\def\zz{1}   

\pgfmathsetmacro\pxe{\ux/(\ux+\uy+\uz)} 
\pgfmathsetmacro\pye{\uy/(\ux+\uy+\uz)}

\pgfmathsetmacro\pxx{(\ux+\vx)/(\ux+\vx+\uy+\uz)} 
\pgfmathsetmacro\pyx{(\uy)/(\ux+\vx+\uy+\uz)}%

\pgfmathsetmacro\pxy{(\ux)/(\ux+\uy+\vy+\uz)} 
\pgfmathsetmacro\pyy{(\uy+\vy)/(\ux+\uy+\vy+\uz)}%

\pgfmathsetmacro\pxz{(\ux)/(\ux+\uy+\vz+\uz)} 
\pgfmathsetmacro\pyz{(\uy)/(\ux+\uy+\vz+\uz)} %

\pgfmathsetmacro\pxxy{(\ux+\vx)/(\ux+\vx+\uy+\vy+\uz)} 
\pgfmathsetmacro\pyxy{(\uy+\vy)/(\ux+\vx+\uy+\vy+\uz)} %

\pgfmathsetmacro\pxxz{(\ux+\vx)/(\ux+\vx+\uy+\vz+\uz)} 
\pgfmathsetmacro\pyxz{(\uy)/(\ux+\vx+\uy+\vz+\uz)} %

\pgfmathsetmacro\pxyz{(\ux)/(\ux+\vy+\uy+\vz+\uz)} 
\pgfmathsetmacro\pyyz{(\uy+\vy)/(\ux+\vy+\uy+\vz+\uz)} %

\pgfmathsetmacro\pxX{(\ux+\vx)/(\ux+\vx+\vy+\uy+\vz+\uz)} 
\pgfmathsetmacro\pyX{(\uy+\vy)/(\ux+\vx+\vy+\uy+\vz+\uz)} %

\pgfmathsetmacro\pze{\zz*(1-\pxe/\xx-\pye/\yy)}
\pgfmathsetmacro\pzx{\zz*(1-\pxx/\xx-\pyx/\yy)}
\pgfmathsetmacro\pzy{\zz*(1-\pxy/\xx-\pyy/\yy)}
\pgfmathsetmacro\pzz{\zz*(1-\pxz/\xx-\pyz/\yy)}
\pgfmathsetmacro\pzxy{\zz*(1-\pxxy/\xx-\pyxy/\yy)}
\pgfmathsetmacro\pzxz{\zz*(1-\pxxz/\xx-\pyxz/\yy)}
\pgfmathsetmacro\pzyz{\zz*(1-\pxyz/\xx-\pyyz/\yy)}
\pgfmathsetmacro\pzX{\zz*(1-\pxX/\xx-\pyX/\yy)}

\coordinate (O)  at (0,0,0);
\coordinate (A)  at (\xx,0,0);
\coordinate (B)  at (0,\yy,0);
\coordinate (C)  at (0,0,\zz);
\coordinate (Pe)  at (\pxe,\pye,\pze);
\coordinate (Px)  at (\pxx,\pyx,\pzx);
\coordinate (Py)  at (\pxy,\pyy,\pzy);
\coordinate (Pz)  at (\pxz,\pyz,\pzz);
\coordinate (Pxy)  at (\pxxy,\pyxy,\pzxy);
\coordinate (Pxz)  at (\pxxz,\pyxz,\pzxz);
\coordinate (Pyz)  at (\pxyz,\pyyz,\pzyz);
\coordinate (PX)  at (\pxX,\pyX,\pzX);

\draw[gray, dotted, very thin] (1,0,0) -- (Pe);
\draw[gray, dotted, very thin] (1,0,0) -- (Py);
\draw[gray, dotted, very thin] (1,0,0) -- (Pz);
\draw[gray, dotted, very thin] (1,0,0) -- (Pyz);
\draw[gray, dotted, very thin] (0,1,0) -- (Pe);
\draw[gray, dotted, very thin] (0,1,0) -- (Px);
\draw[gray, dotted, very thin] (0,1,0) -- (Pz);
\draw[gray, dotted, very thin] (0,1,0) -- (Pxz);
\draw[gray, dotted, very thin] (0,0,1) -- (Pe);
\draw[gray, dotted, very thin] (0,0,1) -- (Py);
\draw[gray, dotted, very thin] (0,0,1) -- (Px);
\draw[gray, dotted, very thin] (0,0,1) -- (Pxy);
\draw[gray!50!black,thick,fill=gray!50,fill opacity=0.1] (A) -- (B) -- (C) -- cycle;
\fill (A)  circle (.2pt) node[below left] {$x$};
\fill (B)  circle (.2pt) node[below right] {$y$};
\fill (C)  circle (.2pt) node[above] {$z$};
\fill[red] (Pe)  circle (.3pt) node[right] {\tiny{$\rho(\emptyset)$}};
\fill[red] (Px)  circle (.3pt) node[left] {\tiny{$\rho(x)$}};
\fill[red] (Py)  circle (.3pt) node[right] {\tiny{$\rho(y)$}};
\fill[red] (Pz)  circle (.3pt) node[above] {\tiny{$\rho(z)$}};
\fill[red] (Pxy)  circle (.3pt) node[below] {\tiny{$\rho(xy)$}};
\fill[red] (Pyz)  circle (.3pt) node[above] {\tiny{$\rho(yz)$}};
\fill[red] (Pxz)  circle (.3pt) node[above] {\tiny{$\rho(xz)$}};
\fill[red] (PX)  circle (.3pt) node[left] {\tiny{$\rho(X)$}};
\fill (CE)  circle (.2pt) node[below left] { };
\node at (CE) {\tiny\(\star\)};
\end{tikzpicture}
    \caption{Special Cases: The left figure illustrates F-Luce for $(u,\bar{v})$: variable $u$ and constant $\bar{v}$. The middle figure demonstrates F-Luce for $(\bar{u},v)$: Constant $u$ and variable $\bar{v}$. Finally, the figure on the right exhibits F-Luce for $(u,\lambda u)$. In all three figures, \(\star\) indicates the uniform choice probabilities.}
    \label{fig:special}
\end{figure}
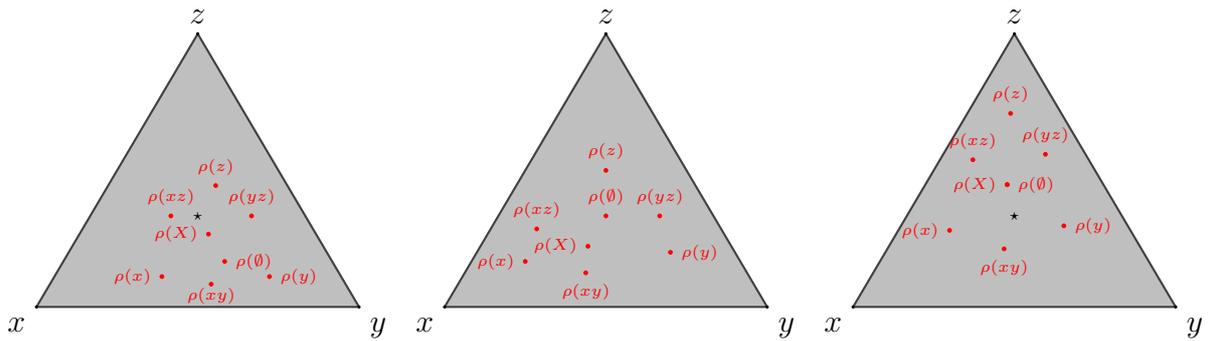


\subsection{Behavioral Characterizations for F-Luce}

In this section, we discuss the behavioral implications of the F-Luce model. These behavioral postulates will allow us to test the model and characterize the empirical content of F-Luce. Since we extend the classical Luce model, as anticipated, some version of Luce's IIA appears in our characterization. Recall that Luce's IIA says the odds of choosing one alternative over another do not depend on the feasible set.

\begin{axxm}[Luce-IIA]\label{Axiom:Luce-IIA}
For $x,y \in F \cap F' $,   $$\frac{ \rho(x,F)}{ \rho(x,F') }=\frac{ \rho(y,F)}{ \rho(y,F') } $$ 
\end{axxm}

Since the F-Luce model is similar to the standard model except that the weights are marked up by a utility $v(x)$ if it is framed, \nameref{Axiom:Luce-IIA} is a necessary condition. However, it is not sufficient. The condition is silent over the choice data for non-framed alternatives. Our model also satisfies the same condition for non-framed alternatives. However, it is still not sufficient, since it does not connect the choice data between framed and non-framed alternatives. The model also requires the same condition between framed and non-framed alternatives. Overall, the Luce-IIA condition must hold as long as $x$ and $y$ do not change their status across $F$ and $F'$. The following axiom imposes that requirement.

\begin{axxm}[Strong Luce-IIA]\label{Axiom:Strong Luce-IIA}
For $x,y \notin F \Delta F' $,   $$\frac{ \rho(x,F)}{ \rho(x,F') }=\frac{ \rho(y,F)}{ \rho(y,F') } $$ where $\Delta$ is the symmetric difference operator.
\end{axxm}

One might suspect that this axiom alone is sufficient for F-Luce. It turns out that we need an additional behavioral postulate. Note that our model assumes framing has a negative impact on non-framed alternatives. Any non-framed item is chosen weakly more when the set of framed items gets smaller. We call this property F-Regularity.

\begin{axxm}[F-Regularity]\label{Axiom:F-Regularity} For $x \notin F$, $\rho(x,F) \leq \rho(x, F \setminus y)$.\end{axxm}

Notice that this axiom is trivially satisfied in the standard Luce model since both sides of the inequality are equal to zero. With this axiom, we are able to state the characterization result.

\begin{thm}\label{Thm:F-Luce}
Let $|X| \geq 3$ and $\mathcal{D}$ include all frames with $|F|\leq 2$. Then, $\rho$ has an F-Luce representation if and only if $\rho$ satisfies \nameref{Axiom:Strong Luce-IIA} and \nameref{Axiom:F-Regularity}.
\end{thm}

Theorem \ref{Thm:F-Luce} works under a very mild data requirement: the size of the frame is less than or equal to 2. One caveat is that Theorem \ref{Thm:F-Luce} requires that there are at least three available alternatives. The reason is that \nameref{Axiom:Strong Luce-IIA} puts no restriction on choice data if there are only two available alternatives. One can define another property, e.g., for $x,y \in F$, $\frac{\rho(x,F)}{\rho(y,F)}\frac{\rho(x,F\setminus\{x,y\})}{\rho(y,F\setminus\{x,y\})}=\frac{\rho(x,F\setminus y)}{\rho(y,F \setminus y)}\frac{\rho(x,F\setminus x)}{\rho(y,F \setminus x)}$. One can show that this property, along with \nameref{Axiom:Strong Luce-IIA} and \nameref{Axiom:F-Regularity}, are the necessary and sufficient conditions for F-Luce with $|X|\geq2$.

We finish this section by establishing the relationships between F-Luce and FRUM. It is a well-known fact that the Luce model is a special case of the random utility model. In the frame environment, one might wonder whether F-Luce belongs to FRUM. It turns out that it does, which we state in the following (see also \Cref{fig:FLUCE-FRUM}).

\begin{remark}\label{Prop: F-Luce is F-RUM}  Every choice rule $\rho$ with an F-Luce representation has an FRUM representation. 
\end{remark}

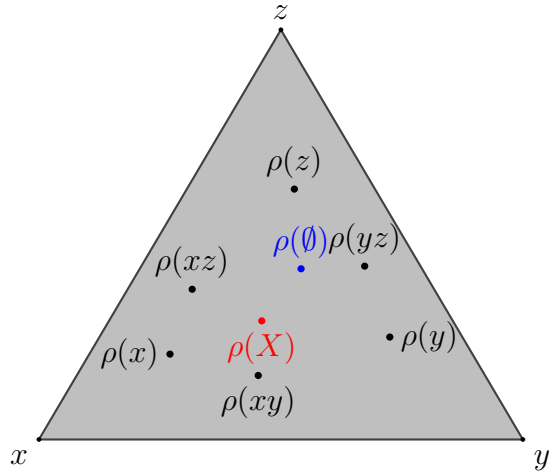
\begin{figure}[h!]
    \centering
\begin{tikzpicture}[line cap=round,line join=round,             3d view={135}{45},        scale=4.5]


\def\ux{1.5}   
\def\uy{2}   
\def\uz{2.5}   

\def\vx{6}   
\def\vy{4}   
\def\vz{3}   
\def\xx{1}   
\def\yy{1}   
\def\zz{1}   

\pgfmathsetmacro\pxe{\ux/(\ux+\uy+\uz)} 
\pgfmathsetmacro\pye{\uy/(\ux+\uy+\uz)}

\pgfmathsetmacro\pxx{(\ux+\vx)/(\ux+\vx+\uy+\uz)} 
\pgfmathsetmacro\pyx{(\uy)/(\ux+\vx+\uy+\uz)}%

\pgfmathsetmacro\pxy{(\ux)/(\ux+\uy+\vy+\uz)} 
\pgfmathsetmacro\pyy{(\uy+\vy)/(\ux+\uy+\vy+\uz)}%

\pgfmathsetmacro\pxz{(\ux)/(\ux+\uy+\vz+\uz)} 
\pgfmathsetmacro\pyz{(\uy)/(\ux+\uy+\vz+\uz)} %

\pgfmathsetmacro\pxxy{(\ux+\vx)/(\ux+\vx+\uy+\vy+\uz)} 
\pgfmathsetmacro\pyxy{(\uy+\vy)/(\ux+\vx+\uy+\vy+\uz)} %

\pgfmathsetmacro\pxxz{(\ux+\vx)/(\ux+\vx+\uy+\vz+\uz)} 
\pgfmathsetmacro\pyxz{(\uy)/(\ux+\vx+\uy+\vz+\uz)} %

\pgfmathsetmacro\pxyz{(\ux)/(\ux+\vy+\uy+\vz+\uz)} 
\pgfmathsetmacro\pyyz{(\uy+\vy)/(\ux+\vy+\uy+\vz+\uz)} %

\pgfmathsetmacro\pxX{(\ux+\vx)/(\ux+\vx+\vy+\uy+\vz+\uz)} 
\pgfmathsetmacro\pyX{(\uy+\vy)/(\ux+\vx+\vy+\uy+\vz+\uz)} %

\pgfmathsetmacro\pze{\zz*(1-\pxe/\xx-\pye/\yy)}
\pgfmathsetmacro\pzx{\zz*(1-\pxx/\xx-\pyx/\yy)}
\pgfmathsetmacro\pzy{\zz*(1-\pxy/\xx-\pyy/\yy)}
\pgfmathsetmacro\pzz{\zz*(1-\pxz/\xx-\pyz/\yy)}
\pgfmathsetmacro\pzxy{\zz*(1-\pxxy/\xx-\pyxy/\yy)}
\pgfmathsetmacro\pzxz{\zz*(1-\pxxz/\xx-\pyxz/\yy)}
\pgfmathsetmacro\pzyz{\zz*(1-\pxyz/\xx-\pyyz/\yy)}
\pgfmathsetmacro\pzX{\zz*(1-\pxX/\xx-\pyX/\yy)}

\coordinate (O)  at (0,0,0);
\coordinate (CE)  at (1/3,1/3,1/3);
\coordinate (A)  at (\xx,0,0);
\coordinate (B)  at (0,\yy,0);
\coordinate (C)  at (0,0,\zz);
\coordinate (Pe)  at (\pxe,\pye,\pze);
\coordinate (Px)  at (\pxx,\pyx,\pzx);
\coordinate (Py)  at (\pxy,\pyy,\pzy);
\coordinate (Pz)  at (\pxz,\pyz,\pzz);
\coordinate (Pxy)  at (\pxxy,\pyxy,\pzxy);
\coordinate (Pxz)  at (\pxxz,\pyxz,\pzxz);
\coordinate (Pyz)  at (\pxyz,\pyyz,\pzyz);
\coordinate (PX)  at (\pxX,\pyX,\pzX);


\coordinate (E1)  at (.44,.257,.303);
\coordinate (E2)  at (.204,.257,.539);
\coordinate (E3)  at (.204,.483,.303);

\draw[blue!50!black,thin,fill=blue!80,fill opacity=0.1] (E1) -- (E2) -- (E3) -- cycle;

\coordinate (F1)  at (.345,.376,.279);
\coordinate (F2)  at (.444,.376,.18);
\coordinate (F6)  at (.345,.287,.368);
\coordinate (F5)  at (.398,.236,.366);
\coordinate (F4)  at (.47,.236,.294);
\coordinate (F3)  at (.47,.35,.18);

\draw[red!50!black,thin,fill=red!50,fill opacity=0.5] (F1) -- (F2) -- (F3) -- (F4) -- (F5) -- (F6) -- cycle;

\draw[gray, dotted, very thin] (1,0,0) -- (Pe);
\draw[gray, dotted, very thin] (1,0,0) -- (Py);
\draw[gray, dotted, very thin] (1,0,0) -- (Pz);
\draw[gray, dotted, very thin] (1,0,0) -- (Pyz);
\draw[gray, dotted, very thin] (0,1,0) -- (Pe);
\draw[gray, dotted, very thin] (0,1,0) -- (Px);
\draw[gray, dotted, very thin] (0,1,0) -- (Pz);
\draw[gray, dotted, very thin] (0,1,0) -- (Pxz);
\draw[gray, dotted, very thin] (0,0,1) -- (Pe);
\draw[gray, dotted, very thin] (0,0,1) -- (Py);
\draw[gray, dotted, very thin] (0,0,1) -- (Px);
\draw[gray, dotted, very thin] (0,0,1) -- (Pxy);
\draw[gray!50!black,thick,fill=gray!50,fill opacity=0.1] (A) -- (B) -- (C) -- cycle;
\fill (A)  circle (.2pt) node[below left] {$x$};
\fill (B)  circle (.2pt) node[below right] {$y$};
\fill (C)  circle (.2pt) node[above] {$z$};

\fill[white] (.5, .5,0)  circle (.0005pt) node[below] {$\rho(x)$};

\fill[black] (Px)  circle (.3pt) node[left] {$\rho(x)$};
\fill[black] (Py)  circle (.3pt) node[right] {$\rho(y)$};
\fill[black] (Pz)  circle (.3pt) node[above] {$\rho(z)$};
\fill[black] (Pxy)  circle (.3pt) node[below] {$\rho(xy)$};
\fill[black] (Pyz)  circle (.3pt) node[above] {$\rho(yz)$};
\fill[black] (Pxz)  circle (.3pt) node[above] {$\rho(xz)$};
\fill[blue] (Pe)  circle (.3pt) node[above] {$\rho(\emptyset)$};
\fill[red] (PX)  circle (.3pt) node[below] {$\rho(X)$};
\end{tikzpicture}
    \caption{F-Luce belongs to FRUM:  The figure illustrates all possible choices for $F=\emptyset$ and $F=X$ for given choices for binary and singleton frames consistent with an F-Luce. Note that $\rho(\cdot,X)$ for F-Luce lies in the red irregular hexagon of F-RUM. Similarly, $\rho(\cdot,\emptyset)$ for F-Luce belongs the blue triangle of F-RUM. This illustrates that F-Luce is a part of FRUM. }
    \label{fig:FLUCE-FRUM}
\end{figure}

Also, if one applies the branch-independent construction to the choice data generated by F-Luce, one can also get a unique branch-independent FRUM representation of F-Luce. 


\subsection{Uniqueness and Identification for F-Luce} 

In this section, we discuss the identification of the model. To achieve this, we first establish that the parameters are unique up to a scaling factor.

\begin{prop}[Uniqueness of F-Luce]\label{Thm: F-Luce Uniqueness}
Let $(u_1,v_1)$ and $(u_2,v_2)$ be two F-Luce representations of the same choice data. Then  $u_1=\alpha u_2$ and $v_1=\alpha v_2$ for some $\alpha>0$.
\end{prop}

We now turn to the recovery of the parameters. Notice that the framing effect for an alternative can be ``observed'' only if the alternative is framed on some occasions. Therefore, to retrieve this parameter, we must have data in which each alternative is framed in some sets. 

\begin{prop}[Identification of F-Luce from Singletons]\label{Prop: F-Luce Identification singleton}
Suppose $\rho$ is F-Luce. Let $\mathcal{D}$ include all frames with $|F|\leq 1$, then we can fully identify the parameters of the model.
\end{prop}

To see the applicability of this proposition, let $X=\{x,y,z\}$. Then, Proposition~\ref{Prop: F-Luce Identification singleton} requires that we only need four framed sets, i.e., $\mathcal{D}=\{\emptyset,\{x\},\{y\},\{z\}\}$. Moreover, to illustrate how identification here works, in the proof, we first recover $u$ from the no-frame scenario, i.e., we let $u(x):=\rho(x,\emptyset)$ for every $x \in X$. Moreover, we define
$$v(x):=\rho(x,\{x\})\frac{1-\rho(x,\emptyset)}{1-\rho(x,\{x\})}-\rho(x,\emptyset)$$
Here, recovering $v(x)$ requires the knowledge of a choice scenario where $x$ is framed, which is given by the choice probability under the framed set $\{x\}$. Moreover, since the parameters are unique up to a scaling factor, we rescale $\rho(x,\{x\})$ according to the ratio of the choice probability of other alternatives across the no-frame and framed $x$. Then, we can fully identify the parameters.

Moreover, note that it is possible to obtain a similar result by replacing the singleton frame data with some doubleton data where each alternative needs to be framed at least once in one of the doubletons. This can potentially further reduce the data requirements for identification. For example, suppose that $X=\{x,y,z,w\}$ and $\mathcal{D}=\{\emptyset,\{x,y\}, \{z,w\} \}$. Then, similarly to before, we let $u(x):=\rho(x,\emptyset)$. For $v$'s, we utilize the doubleton data so that $v(x):=\rho(x,\{x,y\})\frac{1-\rho(x,\emptyset)-\rho(y,\emptyset)}{1-\rho(x,\{x,y\})-\rho(y,\{x,y\})}-\rho(x,\emptyset)$, which shares a similar idea as the construction of $v$ above. In general, one can relax the condition: as long as the dataset $\mathcal{D}$ contains sufficient frames that cover the framed status of each alternative, it will help us to fully identify the parameter.

\begin{prop}[Identification of F-Luce]\label{Prop: F-Luce Identification}
Suppose $\rho$ is F-Luce. If $\emptyset \in \mathcal{D}$ and there exists $F_1,...,F_n \in \mathcal{D}$ where $F_i \neq X$ for all $i$ and $\cup_{i=1}^n F_i = X$, we can fully identify the parameters of the model. In particular, we can let $u(x)=\rho(x,\emptyset)$ and
\begin{equation*}
    v(x)=\rho(x,F)\frac{1-\rho(F,\emptyset)}{1-\rho(F,F)} - u(x)
\end{equation*}
for some $F \in \mathcal{D}$ such that $x \in F$.
\end{prop}

Therefore, it demonstrates the broad applicability of the F-Luce model, where the unique parameters can be identified using potentially limited choice data specifications.

\section{Final Remarks}

In this article, we employ the revealed preference method to develop a frame-dependent decision-making model that captures how an agent's choices are influenced by framing. We introduce a novel choice object and examine the implications of the utility maximization paradigm in both deterministic and probabilistic settings. These models hold promise for economic contexts where the effects of framing are of interest.

\textsc{Future Directions:} Several avenues exist for extending the present work. First, while the utility maximization paradigm provides a robust ``benchmark'' for understanding choice, it often fails to account for certain behavioral and psychological anomalies (e.g., limited consideration, status quo bias, behavioral search, satisficing, temptation, and so on) observed in empirical evidence. To address this, our framework could be expanded to accommodate these irregularities both in deterministic and probabilistic domains. Second, a generalization of the F-Luce model warrants exploration. Although the classical Luce model is widely recognized for its simplicity, it is often deemed ``too restrictive.'' Consequently, alternative parametric models, such as Nested Logit, Dogit, Weighted Linear Model, and Probit model, have been proposed. Similarly, alternative formulations of the F-Luce model could prove to be fruitful subjects of investigation.

\textsc{Broader Applications:} Our model is inspired by the concept of framing, which concerns how the presentation of an option or product influences consumer perception. However, the framework can be extended to scenarios where a product undergoes an actual enhancement due to its framing. For instance, one can apply it to capture consumer impatience through the lens of shipping policies: $x \in F$ can denote an option that benefits from a same-day shipping policy, whereas $x \notin F$ indicates an option subject to standard shipping. In this context, $v(x)$ quantifies the utility increase from faster delivery. Similarly, the framework can be applied to the study of discounts: $x \in F$ signifies a product offered at 15\% off, while $x \notin F$ means the product is at its regular price, with $v(x)$ capturing the perceived utility gain from the discount, capturing any non-quasilinear price effect on utility. While the framework is particularly well-suited for addressing framing-related issues, its versatility suggests broader applicability across various domains. By offering insights into diverse economic phenomena, it opens up new possibilities for research and analysis.

\newpage

\singlespacing

\section{Appendix}\label{Appendix: Supplement material}
In this Appendix, we will provide proofs of the results provided in the main text. Moreover, in the proof, for ease of illustration, we focus on the ordinal version of the model . In particular, we translate $u$ and $v$ into a preference ordering (i.e., a linear order) $\succ$ on $X \cup X^*$, where $X^*$ contains typical element $x^*$ which represents the framed version of $x$. Since $v \geq 0$, for each alternative $x$, it must be that $x^* \succ x$.

\begin{defn}
A deterministic choice rule $c$ has a frame-dependent preference  representation on $\mathcal{D}$ if there exists a preference $\succ$ on $X \cup X^*$  such that  \begin{align*}
    c(R)=\{x \in X| \text{ either }& x \in F \text{ and } x^* \succ \sigma   \text{ for every }  \sigma \in F^* \cup F^c \setminus x^*, \\
    \text{ or } & x \notin F \text{ and } x \succ \sigma \ \text{ for every } \sigma \in F^* \cup F^c \setminus x\}
\end{align*} for all $R \in \mathcal{D}$.
\end{defn}

Note that each $\succ$ over $X\cup X^*$ can be mapped into a pair of utility function $(u,v)$, and vice versa. To give an example, suppose we have $x^* \succ x \succ y^* \succ y$. Then, any $(u,v)$ with $(u+v)(x) > u(x) >(u+v)(y)>u(y)$ will produce the same choice function.

\begin{center}
    \subsubsection*{\textbf{Proof of Theorem~\ref{Thm: Deterministic}}}
\end{center}

\begin{proof}
We would prove the theorem using the ordinal version of the model as discussed in the beginning of the Appendix. We use $\sigma$ as an arbitrary element of $X \cup X^*$.   We define a binary relation $P$ over  $X \cup X^*$. In particular, we let \begin{align*}
    (x^*,y^*) \in P &\text{ if } c(F)=x \text{ and } x,y \in F\\
    (x,y^*) \in P &\text{ if } c(F)=x \text{ and } x \notin F, y \in F \\
    (x^*,y) \in P & \text{ if } c(F)=x \text{ and } x \in F , y \notin F\\
    (x,y) \in P &\text{ if } c(F)=x \text{ and } x,y \notin F
\end{align*}



\begin{claim}\label{claim:P-asymmetric}
 $P$ is asymmetric.
\end{claim}

\begin{proof}
For two distinct $\sigma_1$ and $\sigma_2$, if $(\sigma_1,\sigma_2) \in P$ then $(\sigma_2,\sigma_1) \notin P$. Assume not,  there exists $F_1$ and $F_2$ such that $x=c(F_1)$ and $y=c(F_2)$. There are four cases depending on whether $x$ or $y$ is framed. If $x$ is framed, then $\sigma_1=x^*$, $x$ otherwise. Similarly, $\sigma_2$ is $y$ if it is not framed, $y^*$ otherwise. 

First we assume that $x,y \in F_1 \cap F_2$, then by \nameref{Axiom:IIFA(1)}, we must have $x=c( \{x,y\})=y$, a contradiction. Now assume  $x \in F_1 \cap F_2$ and $y \notin F_1 \cup F_2$, then \nameref{Axiom:IIFA(1)} implies  $x=c( \{x\})$ whereas \nameref{Axiom:IIFA(2)} implies  $y=c( \{x\})$, a contradiction. The third case where $x \notin F_1 \cup F_2$ and $y \in F_1 \cap F_2$ is almost identical to the second case. Finally, assume $x,y \notin F_1 \cup F_2$. Then applying \nameref{Axiom:IIFA(2)} twice yields  $x=c( \emptyset)=y$, a contradiction. Therefore, $P$ is asymmetric.
\end{proof}

Let $a$ be $c(\emptyset)$. Given the definition $P$, we reveal that $a$ is preferred to any alternative in $X \setminus a$. Moreover, we cannot reveal the relative ranking of these alternatives in $X \setminus a$. Formally, for $x,y \neq a$, $a P x $ and $a P y$ then $(x,y) \notin P$ or $(y,x) \notin P$. Notice that $a$ might be revealed to be better than alternatives in $X^*$. Especially, if $a=c(F)$, then $a P y^* $ for all $y \in F$. The same intuition applies to these alternatives too. Hence,  $P$ is incomplete for the lower contour set of $a$.

On the other hand, $P$ is complete in the upper contour set of $a$. To show this, assume $x^*Pa$ and $y^*Pa$ for $x,y \neq a$. In other words, there exist $F \ni x$ and $F' \ni y$ such that $x=c(F)$ and $y=c(F')$ where $a \notin F \cup F'$. If $x \in F'$ or $y \in F$, we would reveal $y^* P x^*$ or $x^* P y^*$, respectively. Assume not. Then consider $\{x,y\}$ as the framed set. First,  $c(\{x,y\}) \in \{x,y\}$. If not, $a=c(\{x\})$ by \nameref{Axiom:IIFA(2)}. Then by \nameref{Axiom:IIFA(1)}, $x$ cannot be chosen from $c(F)$. Hence, either $y^* P x^*$ or $x^* P y^*$, $P$ is complete for the upper contour set of $a$.

\begin{claim}\label{claim:RP-transitive}If $x^* P y^*$ and $x^* \neq a$ then $x^*Pa$.\end{claim}
\begin{proof}
$x^* P y^*$ implies that there exists $F$ such that $x,y \in F$ and $x=c(F)$. Since $\{x\} \subseteq F$, \nameref{Axiom:IIFA(1)} yields that $c(\{x\})=x$ implies  $x^*Pa$.
\end{proof}

Note that our definition of $P$ does not relate $a$ and $a^*$ since it requires two distinct alternatives. However, it is possible that we can conclude that  $a^* P a$. To see this, assume $a=c(F)$ where $a,x \in F$ and $x=c(F')$ where $a \notin F'$. The former implies $a^* P x^* $ and the former yields $x^* P a$. Hence we infer that $a$ is ranked strictly higher when framed. 

\begin{claim}\label{claim:transitive} $P$ is transitive in the strict upper contour set set of $a$.\end{claim}
\begin{proof} Assume that $x^* P y^* P z^*$ for three distinct $x,y$, and $z$. There exists $F$ and $F'$ such that (i) $x=c(F)$ and $y=c(F')$, and ii) $\{x,y\} \subseteq F$ and $\{y,z\} \subseteq F'$. \nameref{Axiom:IIFA(1)} implies that $c(\{x,y\})=x$ and $c(\{y,z\})=y$. Then we consider $c(\{x,y,z\})$. We must have $c(\{x,y,z\})\in \{x,y,z\}$ by \nameref{Axiom:IIFA(2)}. It cannot be $z$ since \nameref{Axiom:IIFA(1)} implies $z=c(\{y,z\})$. It cannot be $y$ since \nameref{Axiom:IIFA(1)}  implies $y=c(\{x,y\})$ which is a contradiction. Hence, we have $x=c(\{x,y,z\})$, which gives $x^* Pz^*$.\end{proof}

Take any completion of the transitive closure of $P$ such that $x^*P x$ for $x$ and call it $\succ$. Note that the transitive closure would immediately imply  $x^*P x$ if there exists $F\neq \emptyset$ such that $c(F)=x$ by using Claim~\ref{claim:RP-transitive} and the definition of $P$ by the data $c(\emptyset)=a$. For $x$ that is never chosen, $x^*P x$ is trivial.

It is routine to show that the representation holds by choosing $u$ and $v$ appropriately. Note that $v\geq 0$ by the fact of $x^* \succ x$.  \end{proof}

\begin{center}
    \subsubsection*{ \textbf{Proof of Proposition~\ref{Thm: PFdeterministic uniqueness}}}
\end{center}

\begin{proof}
Suppose $(u_1,v_1)$ and $(u_2,v_2)$ represent the same choice rule. For i), note that it is immediate that $\argmax_{x \in X} u_1(x)=c(\emptyset)=\argmax_{x \in X} u_2(x)$. We call it $a$. For ii), suppose not, there exists $b$ such that $u_1(a)>u_1(b)+v_1(b)$ and $u_2(b)+v_2(b)>u_2(a)$. Then, we know that $c_{(u_1,v_1)}(\{b\})=a \neq b=c_{(u_2,v_2)}(\{b\})$. Contradiction arises. For iii), suppose not, there exists $x,y \in X$  such that $u_1(x)+v_1(x) > u_1(y)+v_1(y)>u_1(a)$ but  $  u_2(y)+v_2(y)>u_2(x)+v_2(x)>u_2(a)$. Then, we have $c_{(u_1,v_1)}(\{x,y\})=x \neq y =c_{(u_2,v_2)}(\{x,y\})$. Contradiction arises.
\end{proof}

\begin{center}
   \subsubsection*{  \textbf{Proof of Lemma~\ref{Thm: Falmagne 3 general}}}
\end{center}

To prove Theorem 5, we need first to prove the following Lemma.

\begin{lem}\label{Thm: Falmagne 3 general}
For  $F \subseteq X$ and choice rule $\rho$, 
$$\sum_{a \in F} q_\rho(a, F) + \sum_{a \notin F} y_\rho(a, F) = \sum_{b \notin F } q_\rho(b, F \cup {b})$$
\end{lem}

\begin{proof}
We prove by strong induction by ``stepping down.'' For $F=X \setminus \{x\}$, we have RHS $= q(x,X)= \rho(x,X)$. Then,
\begin{align*}
\text{LHS}  &= \sum_{a \in X \setminus \{x\}} q(a, X \setminus \{x\}) +  \sum_{a \notin X \setminus \{x\}} y(a, X \setminus \{x\}) \\
&= \sum_{a \in X \setminus \{x\}} q(a, X \setminus \{x\}) + \rho(x, X \setminus \{x\}) \\
&=  \sum_{a \in X \setminus \{x\}} \bigg[ \rho (a, X \setminus \{x\}) - \rho(a,X) \bigg] + \rho(x, X \setminus \{x\}) \\
&= \sum_{a \in X }  \rho (a, X \setminus \{x\}) - \sum_{a \in X \setminus \{x\}} \rho(a,X) \\
&= 1 - \sum_{a \in X \setminus \{x\}} \rho(a,X) = \rho(x,X)
\end{align*}
Hence, $LHS=RHS$. The first step is proven. Suppose that equality holds for size of $k+1,k+2,...N-1$. Let $|F|=k$, 
\begin{align*}
\text{LHS}-\text{RHS}
    =&\sum_{a \in F} q(a, F) + \sum_{a \notin F} y(a, F) - \sum_{b \notin F } q(b, F \cup {b})\\
 = &\sum_{a \in F} \bigg(\rho(a,F)- \sum_{B \supset F} q(a,B)\bigg) +  \sum_{a \notin F} \bigg(\rho(a,F)- \sum_{a \notin B \supset F} y(a,B)\bigg) - \sum_{b \notin F } q(b, F \cup {b})\\
 =& \sum_{a \in X} \rho(a,F) -\bigg[ \sum_{a \in F} \sum_{B \supset F} q(a,B) +  \sum_{a \notin F} \sum_{a \notin B \supset F} y(a,B) + \sum_{b \notin F } q(b, F \cup {b}) \bigg] 
 \end{align*}
 
Since $\sum\limits_{a \in X} \rho(a,F)=1$, it remains to show that the latter term in the above expression equals 1. We denote $\mathcal{D}_R(i)$ as the collection of supersets of $F$ with $i$ element. Hence, we can rewrite 
\begin{align*}
  &   \sum_{a \in F}\sum_{B \supset F} q(a,B) +  \sum_{a \notin F} \sum_{a \notin B \supset F} y(a,B) + \sum_{b \notin F } q(b, F \cup {b}) \\
=& \sum_{i=|F|+1}^{N} \sum_{B \in \mathcal{D}_R(i)} \bigg[ \sum_{a \in F} q(a,B) + \sum_{a\notin B} y(a,B) \bigg] + \sum_{B \in \mathcal{D}_R(|F|+1)} \sum_{a\notin F} q(a,B) \tag*{By rearrangement}\\
=& \sum_{i=|F|+2}^{N} \sum_{B \in \mathcal{D}_R(i)} \bigg[ \sum_{a \in F} q(a,B) + \sum_{a\notin B} y(a,B) \bigg] + \sum_{B \in \mathcal{D}_R(|F|+1)} \bigg[ \sum_{a \in B} q(a,B) + \sum_{a\notin B} y(a,B) \bigg] \tag*{By taking $i=|F|+1$ from the 1st term and summing it to the second term}\\
=& \sum_{i=|F|+2}^{N} \sum_{B \in \mathcal{D}_R(i)} \bigg[ \sum_{a \in F} q(a,B) + \sum_{a\notin B} y(a,B) \bigg] + \sum_{B \in \mathcal{D}_R(|F|+1)} \sum_{a \notin B } q(a, B \cup {a}) \tag*{By induction hypothesis}\\
=& \sum_{i=|F|+2}^{N} \sum_{B \in \mathcal{D}_R(i)} \bigg[ \sum_{a \in F} q(a,B) + \sum_{a\notin B} y(a,B) \bigg] + \sum_{B \in \mathcal{D}_R(|F|+2)} \sum_{a\notin F} q(a,B) \tag*{By rearrangement}\\
=& ...\text{(repetitively applying induction hypothesis)}\\
=& \sum_{i=N}^{N} \sum_{B \in \mathcal{D}_R(i)} \bigg[ \sum_{a \in F} q(a,B) + \sum_{a\notin B} y(a,B) \bigg] + \sum_{B \in \mathcal{D}_R(|N|)} \sum_{a\notin F} q(a,B) = \sum_{a \in X} q(a,X) = \sum_{a \in X} \rho(a,X) = 1
\end{align*} where the second-to-last equality is given by the definition of $q$. Hence, it is proven. \end{proof}

\begin{center}
    \subsubsection*{ \textbf{Proof of Theorem~\ref{Thm: F-RUM}}}
\end{center}

\begin{proof}
In this proof, we study the representation through the preference types, so that it connects better with the proof in Theorem 1 and displays the interchangeability of the primitives. Therefore, we consider a distribution, $\mu$, over preference preference types over $X\cup X^*$. 

For the necessity proof, we suppose the data follows the model. We introduce the following notation, for $b \in X$ and $A \subseteq X \setminus b$,    
\begin{align*}
    M_{y}(b,A)&:=\mu\Big(\{\succ| (X\setminus b) \cup  A^* =L_\succ(b)\}\Big)\\
    M_{q}(b,A\cup b)&:=\mu\Big(\{\succ| X \cup A^*=L_\succ(b^*)\}\Big)
\end{align*}where $L_\succ(a)$ is the strict lower contour set of $a$ according to $\succ$.

\begin{claim}\label{Claim: two properties of W neccessity}
For $b\in X$ and $A \subseteq X \setminus b$, \\
i) $M_y(b,A) = y(b,A)$,\\
ii) $M_q(b,A) = q(b,A\cup b)$
\end{claim}

\begin{proof}
We prove by strong induction by ``stepping down.'' For $A= X\setminus \{x\}$, we have, 
\begin{align*}
y(x,X\setminus \{x\})&= \rho(x, X\setminus \{x\})\\
& = \mu  \Big(\{\succ| (X\setminus \{x\}) \cup  (X\setminus \{x\})^* =L_\succ(x)\}\Big)\\
&= M_y(x,X\setminus \{x\})
\end{align*}
So, i) is true for size of $A$ equals to $|X|-1$. Suppose i) is true for size of  $k+1$, $k+2$,...,$|N|-1$. Let $|A|=k$,     \begin{align*}
y(x,A) & =\rho(x,A) - \sum\limits_{x\notin B \supset A} y(x,B)\\
&=\sum\limits_{x \notin B \supseteq A} M_y(x,B) - \sum\limits_{x\notin B \supset A} M_y(x,B)  \\
&=M_y(x,A) \end{align*} where the second-to-last equality is given by Definition and induction hypothesis. Hence, the proof is complete for i). For ii), we prove again by strong induction by ``stepping down.''  For $A= X\setminus \{x\}$, we have 
\begin{align*}
q(x,X)&= \rho(x, X)\\
&= \mu   \Big(\{\succ| X \cup  (X\setminus \{x\})^* =L_\succ(x)\}\Big)\\
&= M_q(x,X\setminus \{x\})\end{align*} So, ii) is true for size of $A$ equals to $|X|-1$. Suppose ii) is true for size of  $k+1$, $k+2$,...,$|N|-1$. Let $|A|=k$.
\begin{align*}
q(x,A) & =\rho(x,A) - \sum\limits_{x\notin B \supset A} q(x,B)\\
&=\sum\limits_{x \notin B \supseteq A} M_q(x,B) - \sum\limits_{x\notin B \supset A} M_q(x,B)  \\
&=M_q(x,A)\end{align*} where the second-to-last equality is given by Definition and induction hypothesis. Hence, the proof is complete for ii).
\end{proof}

From this claim, we immediately show   \nameref{Axiom: non-negativity of BM}, since $\mu$ is by definition positive.

For the sufficiency proof, for ease of notation, we introduce a new notation. We provide a recursive construction here, which requires us to construct the preference $\succ$. Therefore, we denote each $\succ$ with a sequence of alternatives (a string of text). For example, $x^* \succ x \succ y^*\succ y$ will be simply denoted by $x^*xy^*y$. For any $F^* \subseteq X^*$, we write $\Pi_{F^*}$ for the set of $|F^*|!$ permutations on $F$, with typical element $\pi_{F^*}$. We will note $\pi_{F^*}\sigma$ as the extended sequence for some $\sigma \in X\cup X^*$. we first construct $G(a^*):=q(a,X)$ for all $a^* \in X^*$. Then, we construct, recursively, 

\begin{align*}
    G(\pi_{F^*}\sigma)= \begin{dcases}
    G(\pi_{F^*}) \frac{y(z, X\setminus F)}{\sum_{ \pi_{F^*}' \in \Pi_{F^*}}  G(\pi_{F^*})} & \text{ if } \sigma = z \in F\\
    G(\pi_{F^*}) \frac{q(z, X\setminus F)}{\sum_{ \pi_{F^*}' \in \Pi_{F^*}}  G(\pi_{F^*})} & \text{ if } \sigma = z^* \in (X\setminus F)^*\\
    \end{dcases}
\end{align*}

Note that $G\geq0$ by \nameref{Axiom: non-negativity of BM}. We then prove the following properties for $G$.

\begin{claim}\label{Claim:Two properties of G} For $F \subseteq X$, we have,\\
i) $\sum_{ \pi_{F^*} \in \Pi_{F^*}} G(\pi_{F^*}\sigma) = y(x, X \setminus F)  $ for $\sigma =x \in F$, \\
ii)  $\sum_{ \pi_{F^*} \in \Pi_{F^*}} G(\pi_{F^*}\sigma) = q(x, X \setminus F)  $ for $\sigma = x^* \in (X\setminus F)^*$, \\
iii) $\sum_{ \pi_{F^*} \in \Pi_{F^*}}  G(\pi_{F^*}) = \sum_{x\in F} q(x,x\cup (X\setminus F)) $, \\
iv) $G(\pi_{F^*}) = \sum_{\sigma \in F \cup (X\setminus F)^*} G(\pi_{F^*}\sigma) $  for all $\pi_{F^*}$,
\end{claim}

\begin{proof}
Note that i) and ii) are straight forward. For iii), we have 
\begin{align*}
LHS= &\sum_{ \pi_{F^*} \in \Pi_{F^*}}  G(\pi_{F^*})  \\
&= \sum_{ z \in F} \sum_{ \pi_{(F\setminus z)^*} \in \Pi_{(F\setminus z)^*}}  G(\pi_{(F\setminus z)^*z^*})  \\
&= \sum_{ z \in F} q(z,x\cup(X\setminus F)) \\
&=RHS \text{ (By (ii))}
\end{align*}      For iv), we have  \begin{align*}
RHS & = \sum_{\sigma \in F \cup (X\setminus F)^*} G(\pi_{F^*}\sigma) \\
&  = \frac{G(\pi_{F^*})}{\sum_{ \pi_{F^*}' \in \Pi_{F^*}}  G(\pi'_{F^*})} (\sum_{z \in F} y(z,X\setminus F) + \sum_{z\in X\setminus F} q(z,X\setminus F))\\
& = \frac{G(\pi_{F^*})}{\sum_{ \pi_{F^*}' \in \Pi_{F^*}}  G(\pi'_{F^*})} (\sum_{ z \in F} q(z,z\cup(X\setminus F)) \text{ (By Lemma~\ref{Thm: Falmagne 3 general}}) \\         & = G(\pi_{F^*}) = LHS \text{ (By (iii))}    \end{align*}
\end{proof}

We then define each individual weight. We will consider permutation on $X\cup X^*$. We let $\pi$ denote each permutation in $X\cup X^*$. Note that we only consider permutation where $x^*$  comes before $x$. Then, we perform a ``truncation'' on $\pi$. Firstly, we truncate $\pi$ up to the first element in $X$ that appears in the sequence. For example, for $X=\{x,y,z\}$, we will truncate the following sequence in the following way. $$\underbrace{x^*y^*x}_{\pi^t}z^*zy$$ Here, since the utility type space of FRUM is huge and many types are inherently indistinguishable from other,  we mainly assign the weights on a group of types whose behaviors are the same. Then, we allocate the weight evenly across the group. We define,  for every $\pi_{F^*} \in \Pi_{F^*}$ and $x\in X$, for every $\pi$ such that $\pi^t=\pi_{F^*}x$, $$\hat{\mu}_{\pi}:=G(\pi_{F^*}x)*\frac{1}{|\{\pi: \pi^t=\pi_{F^*}x \}|}.$$ Firstly, note that $\hat{\mu}\geq 0$ due to the fact that $G\geq0$. Claim~\ref{Claim:Two properties of G}(iv) ensures that $\hat{\mu}$ is additive. In the following, we introduce the following notation. Let $L_\pi(\sigma)$ be the lower contour set of $\sigma$ according to $\pi$. We define  $\hat{M}_{y}(b,A):=\hat{\mu}\Big(\{\pi| (X\setminus b) \cup  A^* =L_\pi(b)\}\Big)$ and $\hat{M}_{q}(b,A\cup b):=\hat{\mu}\Big(\{\pi| X \cup A^*=L_\pi(b^*)\}\Big)$. Since the weights are constructed, we now check for the choice data. For $x \notin F$,  \begin{align*}
    \hat{\rho}(x,F) &= \sum_{ x\notin  B\supseteq F } \hat{M}_y(x,B) \text{ (by definition)}\\
    &= \sum_{ x\notin  B\supseteq F } \sum_{\pi_{(X \setminus B)^*} \in \Pi_{(X \setminus B)^*}} G(\pi_{(X \setminus B)^*}x)   \text{ (by construction)}\\
    &= \sum_{ x\notin  B\supseteq F } y(x,B) \text{ (By Claim~\ref{Claim:Two properties of G}(i))} \\
    &= \rho(x,F)
\end{align*}

On the other hand, for $x \in F$,  \begin{align*}
    \hat{\rho}(x,F) &= \sum_{ A\supseteq F } \hat{M}_q(x,A)  \text{ (by definition)}\\
    &= \sum_{ A\supseteq F }  \sum_{\pi_{(X \setminus A)^*} \in \Pi_{(X \setminus A)^*}} G(\pi_{(X \setminus A)^*}x) \text{ (by construction)}\\
    &= \sum_{ A\supseteq F } q(x,A) \text{ (By Claim~\ref{Claim:Two properties of G}(ii))}\\
    & = \rho(x,F)
\end{align*}
Hence, the constructed weights explain the data. Since it explains the data, we have  $\sum\limits_{\pi \in \Pi}  \hat{\mu}_{\pi}=1$. The sufficiency proof is complete.
\end{proof}

\begin{center}
     \subsubsection*{\textbf{Proof of Proposition~\ref{Thm: FRUM uniqueness}}}
\end{center}

\begin{proof}
It is proven in Claim~\ref{Claim: two properties of W neccessity}. To see this, using the primitive $\succ$, we defined before Claim~\ref{Claim: two properties of W neccessity} that for $a \in X$ and $A \subseteq X \setminus a$, $M_{y}(b,A):=\mu\Big(\{\succ| (X\setminus b) \cup  A^* =L_\succ(b)\}\Big)$ and $M_{q}(b,A\cup b):=\mu\Big(\{\succ| X \cup A^*=L_\succ(b^*)\}\Big)$ and proved in Claim~\ref{Claim: two properties of W neccessity} that 
For $a \in X$ and $A \subseteq X \setminus a$, i) $M_y(b,A) = y(b,A)$ and ii) $M_q(b,A) = q(b,A\cup b)$. Therefore, for Proposition~\ref{Thm: FRUM uniqueness}(i), we have $y_\rho(b,F)=M_y(b,F)=\mu(\{\succ|(X\setminus b) \cup F^* = L_\succ (b)\})$. In other words, it captures all the types $\succ$ where $b$ is better than all alternatives in $F^*$ (the framed version) and all alternatives in $X\setminus b$ (the regular version), and is worse than all alternatives in $(X\setminus F)^*$ (the framed version). Therefore, it gives  $y_\rho(b,F)=\mu^i\big(\{c \in \mathcal{C}_{\text{FUM}} | \ 
 b=c(F) \  \text{ and }  \   x = c(\{x\}) \ \text{ for all } x \notin F  \})$. Similar arguments can be made for Proposition~\ref{Thm: FRUM uniqueness}(ii) using Claim~\ref{Claim: two properties of W neccessity}(ii). \end{proof}

\begin{center}
    \subsubsection*{ \textbf{Proof of Theorem~\ref{Thm:F-Luce}}}
\end{center}

\begin{proof} The necessity proof is straightforward. In the following, we will adopt three IIA implied by \nameref{Axiom:Strong Luce-IIA}. Namely, framed Luce-IIA, Non-framed Luce-IIA and Cross-status Luce-IIA, which are discussed in the main text.

For the sufficiency, we will prove it with two different cases. The first case is that $|X|\geq 3$. we let $u(x):=\rho(x,\emptyset)$. We designate an anchoring element $z_0$ such that for every $x \in X\setminus z_0$, we let $l(x|F,z_0):=\frac{\rho(x,F)}{\rho(z_0,F)}\rho(z_0,\emptyset)$ for some $F$ such that $x \in F$ and $z_0 \notin F$. Firstly, note that the denominator is greater than 0 by the assumption that choice probability is positive.

\begin{claim}
    For $F\neq X$, $y,z\notin F$ and $x \in F$, $l(x|\{x\},z)=l(x|F,y)$.
\end{claim}

\begin{proof}
    To see this, let $y \notin F$ , by Non-framed Luce-IIA, we have $\rho(y,\{x\})=\frac{\rho(z,\{x\})}{\rho(z,\emptyset)}\rho(y,\emptyset)$, which we substitute into $\frac{\rho(x,\{x\})}{\rho(x,F)}=\frac{\rho(y,\{x\})}{\rho(y,F)}$ (by Cross-status Luce-IIA). By simplifying, we get $l(x|F,z):=\frac{\rho(x,\{x\})}{\rho(z,\{x\})}\rho(z,\emptyset)=\frac{\rho(x,F)}{\rho(y,F)}\rho(y,\emptyset):=l(x|F,y)$.
\end{proof}

Hence, the definition of $l$ is independent of the anchoring element $z_0$. We will write $l(x)$ to replace $l(x|F,z_0)$ for every $x \notin z_0$. Also, we can define $l(z_0)$ which is any other arbitrary anchor.

\begin{claim}
1) For $x,y \notin F$, $\frac{\rho(x,F)}{\rho(y,F)}=\frac{u(x)}{u(y)}$. 2) For $x,y\in F$,  $\frac{\rho(x,F)}{\rho(y,F)}=\frac{l(x)}{l(y)}$.  3) For $x\in F$ and $y \notin F$,  $\frac{\rho(x,F)}{\rho(y,F)}=\frac{l(x)}{u(y)}$.
\end{claim}

\begin{proof}
To see 1) holds, note that Non-framed Luce-IIA implies $\frac{\rho(x,F)}{\rho(y,F)}=\frac{\rho(x,\emptyset)}{\rho(y,\emptyset)}=\frac{u(x)}{u(y)}$. To see 2) holds, note that for $F\neq X$ and some $z \notin F$,  $\frac{\rho(x,F)}{\rho(y,F)}=\frac{\rho(x,F)}{\rho(y,F)}\frac{\rho(z,\emptyset)}{\rho(z,\emptyset)}\frac{\rho(z,F)}{\rho(z,F)}=\frac{l(x)}{l(y)}$. For $F=X$, we use framed Luce-IIA  so that $\frac{\rho(x,X)}{\rho(y,X)}=\frac{\rho(x,\{x,y\})}{\rho(y,\{x,y\})}=\frac{\rho(x,\{x,y\})}{\rho(y,\{x,y\})}\frac{\rho(z,\emptyset)}{\rho(z,\emptyset)}\frac{\rho(z,\{x,y\})}{\rho(z,\{x,y\})}=\frac{l(x)}{l(y)}$ for some $z\notin \{x,y\}$.  To see 3) holds, $\frac{\rho(x,F)}{\rho(y,F)}=\frac{\rho(x,F)}{\rho(y,F)}\frac{\rho(y,\emptyset)}{\rho(y,\emptyset)}=\frac{l(x)}{u(y)}$. The proof is complete.
\end{proof}

Lastly,  for $x \in F$, taking summation over all choice ratio of every elements in $X$ with $\rho(x,F)$ being at the denominator, we have  $\sum_{y \in X}\frac{\rho(y,F)}{\rho(x,F)}=\frac{\sum_{y \in F} l(x) + \sum_{y \in F^c} u(x)}{l(x)} \iff \rho(x,F)= \frac{l(x)}{\sum_{y \in F} l(x) + \sum_{y \in F^c} u(x)}$.  One can apply the same argument analogously to the case that $x \notin F$. Then, we define $v(x):=l(x)-u(x)$. The representations hold if $l(x) \geq u(x)$ for every $x \in X$. To see this, 
We check $\rho(y,\{x\})$ and $\rho(y,\emptyset)$ for some $y \neq x$. It is easy to show that $l(x)\geq u(x)$ if and only if  $\rho(y,\{x\})=\frac{u(y)}{\sum_{z\in X\setminus x} u(z) + l(x)} \geq \frac{u(y)}{\sum_{z\in X\setminus x} u(z) + u(x)}=\rho(y,\emptyset)$
which is given by \nameref{Axiom:F-Regularity}. Hence, the proof is complete. \end{proof}

\begin{center}
    \subsubsection*{ \textbf{Proof of \Cref{Prop: F-Luce is F-RUM}}}
\end{center}

\begin{proof}
It suffices to show that it satisfies \nameref{Axiom: non-negativity of BM}. We first state the following fact, due to the well-known fact that any Luce is RUM.

\begin{fact}
Let $k_i \geq 0$ for $i=1,...,N$, and we denote $\mathcal{N}=\{1,..,N\}$ Then, for any $K>0$, we have
$\sum_{A \subseteq \mathcal{N}} (-1)^{|A|} \frac{1}{K+\sum_{i \in A}k_i} \geq 0$
\end{fact}

We first consider $y(x,F)$.  $y(x,F) = \sum_{x \notin B \supseteq F} (-1)^{|B\setminus F|} \rho(x,B)
           $ $= \sum_{x \notin B \supseteq F} (-1)^{|B\setminus F|} \frac{u(x)}{u(X)+v(B)}
           $ $= u(x)\sum_{x \notin B \supseteq F} (-1)^{|B\setminus F|} \frac{1}{u(X)+v(F)+v(B\setminus F)}
           $ $= u(x)\sum_{A \subseteq X\setminus(F\cup x)} (-1)^{|A|} \frac{1}{u(X)+v(F)+v(A)}$. We use the above fact by letting $u(X)+v(F)=K>0$ with $v(x)\geq 0$. Then, it is immediate that $y(x,F) \geq 0$. Analogously, one can also see that $q(x,F) \geq 0$.
\end{proof}

\begin{center}
   \subsubsection*{  \textbf{Proof of Proposition~\ref{Thm: F-Luce Uniqueness}}}
\end{center}


\begin{proof}
Suppose that $(u_1,v_1)$ and $(u_2,v_2)$ represent the same choice rule. Then, by definition, for every $x \in X$, $\frac{u_1(x)}{u_1(X)} = \rho(x,\emptyset)=\frac{u_2(x)}{u_2(X)}$. Therefore, we get $u_1(x)=\frac{u_1(X)}{u_2(X)}u_2(x)$, where $\frac{u_1(X)}{u_2(X)}>0$ and is independent of $x$. Also, for every $x \in X$, we have $\frac{u_1(z)}{u_1(X)+v_1(x)}=\rho(z,\{x\})=\frac{u_2(z)}{u_2(X)+v_2(x)}$. Hence, by putting $u_1=\alpha u_2$, we have $\frac{\alpha u_2(z)}{\alpha u_2(X)+ v_1(x)}= \frac{u_2(z)}{u_2(X)+v_2(x)}$. Then, we get $v_1(x)=\alpha v_2(x)$.\end{proof}

\begin{center}
    \subsubsection*{ \textbf{Proof of Proposition~\ref{Prop: F-Luce Identification singleton}}}
\end{center}

\begin{proof}
    To derive $v(x)$, note that we have $\rho(x,x)=\frac{v(x)}{v(x)+u(X\setminus x)}$. By putting in $\rho(x,X\setminus x) =u(X\setminus x)$ and re-arrangement, one can get $v(x):=\rho(x,\{x\})\frac{1-\rho(x,\emptyset)}{1-\rho(x,\{x\})}$. To see that it explains choice data. Note that $v(x)$ is derived from $\rho(x,x)$, it is necessarily true that $(u,v)$ explains $\rho(x,x)$. For $y\neq x$, note that
    \begin{align*}
        \rho(y,x)&=\frac{\rho(y,\emptyset)}{\rho(x,x)\frac{1-\rho(x,\emptyset)}{1-\rho(x,x)}+\rho(y,\emptyset)+\rho(X\setminus\{x,y\},\emptyset)}\\
        &=\frac{\rho(y,\emptyset)}{\frac{1-\rho(x,\emptyset)}{1-\rho(x,x)}}\\
        &=(\rho(y,x)+\rho(X\setminus\{x,y\},x))\frac{\rho(y,\emptyset)}{\rho(y,\emptyset)+\rho(X\setminus\{x,y\},\emptyset)}\\
        &=(\rho(y,x)+\rho(X\setminus\{x,y\},x))\frac{\rho(y,x)}{\rho(y,x)+\rho(X\setminus\{x,y\},x)}=\rho(y,x)
    \end{align*}
where the second-to-last step is given by \nameref{Axiom:Strong Luce-IIA}.
\end{proof}

\begin{center}
    \subsubsection*{ \textbf{Proof of Proposition~\ref{Prop: F-Luce Identification}}}
\end{center}
\begin{proof}
    The proof is analogous to the last proposition.
\end{proof}

\newpage

\bibliographystyle{econ-a}
\bibliography{FrameReferences}

\end{document}